\documentclass[a4paper,UKenglish,cleveref, autoref, thm-restate]{lipics-v2021}

\pdfoutput=1 
\hideLIPIcs  
\nolinenumbers

\title{The Equivalence Problem of E-Pattern Languages with Length Constraints is Undecidable} 
\titlerunning{On Pattern Languages with Length Constraints} 

\author{Dirk Nowotka}{Kiel University, Department of Computer Science, Germany}{dn@informatik.uni-kiel.de}{}{}
\author{Max Wiedenhöft}{Kiel University, Department of Computer Science, Germany}{maw@informatik.uni-kiel.de}{}{This work was supported by the DFG project number 437493335.}

\authorrunning{D. Nowotka and M. Wiedenhöft} 

\Copyright{Dirk Nowotka and Max Wiedenhöft}
\ccsdesc[300]{Theory of computation~Formal languages and automata theory}

\keywords{Patterns, Pattern Languages, Length Constraints, Regular Constraints, Decidability, Undecidability, Membership, Inclusion, Equivalence}

\EventEditors{John Q. Open and Joan R. Access}
\EventNoEds{2}
\EventLongTitle{42nd Conference on Very Important Topics (CVIT 2016)}
\EventShortTitle{CVIT 2016}
\EventAcronym{CVIT}
\EventYear{2016}
\EventDate{December 24--27, 2016}
\EventLocation{Little Whinging, United Kingdom}
\EventLogo{}
\SeriesVolume{42}
\ArticleNo{23}

\usepackage{todonotes}
\usepackage{float}
\usepackage{array}
\usepackage{tablefootnote}

\DeclareMathOperator{\Fact}{Fact}

\DeclareMathOperator{\Pref}{Pref}

\DeclareMathOperator{\Suff}{Suff}

\DeclareMathOperator{\var}{var}
\DeclareMathOperator{\enc}{enc}

\def\N{\mathbb{N}}
\def\ta{\mathtt{a}}
\def\tb{\mathtt{b}}

\def\Pat{\mathtt{Pat}}
\def\PatLen{\mathtt{Pat}_{\Sigma,\mathcal{C}_{Len}}}
\def\PatReg{\mathtt{Pat}_{\Sigma,\mathcal{C}_{Reg}}}
\def\PatRegLen{\mathtt{Pat}_{\Sigma,\mathcal{C}_{Reg,Len}}}

\def\ValC{\mathtt{ValC}}
\def\ValCP{\mathtt{ValC}^{\rightleftarrows}}

\newtheorem{question}{Question}

\newcommand{\al}{\operatorname{alph}}

\newcolumntype{C}{>{\hspace{3pt}}c<{\hspace{3pt}}}

\begin{document}

\maketitle

\begin{abstract}
Patterns are words with terminals and variables. The language of a pattern is the set of words obtained by uniformly substituting all variables with words that contain only terminals. Length constraints restrict valid substitutions of variables by associating the variables of a pattern with a system (or disjunction of systems) of linear diophantine inequalities. Pattern languages with length constraints contain only words in which all variables are substituted to words with lengths that fulfill such a given set of length constraints. We consider membership, inclusion, and equivalence problems for erasing and non-erasing pattern languages with length constraints. 

Our main result shows that the erasing equivalence problem - one of the most prominent open problems in the realm of patterns - becomes undecidable if length constraints are allowed in addition to variable equality. 

Additionally, it is shown that the terminal-free inclusion problem, a prominent problem which has been shown to be undecidable in the binary case for patterns without any constraints, is also generally undecidable for all larger alphabets in this setting.

Finally, we also show that considering regular constraints, i.e., associating variables also with regular languages as additional restrictions together with length constraints for valid substitutions, results in undecidability of the non-erasing equivalence problem. This sets a first upper bound on constraints to obtain undecidability in this case, as this problem is trivially decidable in the case of no constraints and as it has unknown decidability if only regular- or only length-constraints are considered.
\end{abstract}

\newpage

\section{Introduction}
A \emph{pattern} is a finite word consisting only of symbols from a finite set of \emph{letters} $\Sigma = \{\ta_1,...,\ta_\sigma\}$, also called \emph{terminals}, and from an infinite set of \emph{variables} $X = \{x_1, x_2, ... \}$ such that we have $\Sigma \cap X = \emptyset$. It is a natural and compact device to define formal languages. From patterns, we obtain words consisting only of terminals using a \emph{substitution} $h$, a terminal preserving morphism that maps all variables in a pattern to words over the terminal alphabet. The \emph{language} of a pattern is the set of all words obtainable using arbitrary substitutions. 

We differentiate between two kinds of substitutions. In the original definition of patterns and pattern languages introduced by Angluin \cite{DBLP:journals/jcss/Angluin80}, only words obtained by \emph{non-erasing substitutions} are considered. Here, all variables are required to be mapped to non-empty words. The resulting languages are called \emph{non-erasing (NE) pattern languages}. Later, so called \emph{erasing-/extended-} or just \emph{E-pattern languages} have been introduced by Shinohara \cite{Shinohara1983}. Here, variables may also be substituted by the empty word $\varepsilon$. Consider, for example, the pattern $\alpha := x_1 \ta x_2 \tb x_1$. Then, if we map $x_1$ to $\ta\tb$ and $x_2$ to $\tb\tb\ta\ta$ using a substitution $h$, we obtain the word $h(\alpha) = \ta\tb\ta\tb\tb\ta\ta\tb\ta\tb$. Considering the E-pattern language of $\alpha$, we could also map $x_1$ to the empty word and obtain any word in the language $\{\ta\}\cdot\Sigma^*\cdot\{\tb\}$.

Due to its practical and simple definition, patterns and their corresponding languages occur in numerous areas in computer science and discrete mathematics. These include, for example, unavoidable patterns~\cite{Jiang1994,lothaire1997}, algorithmic learning theory~\cite{DBLP:journals/jcss/Angluin80,FERNAU201844,Shinohara1995}, word equations~\cite{lothaire1997}, theory of extended regular expressions with back references~\cite{FREYDENBERGER20191}, or database theory~\cite{FreydenbergerP21,SchmidSchweikardtPODS2022}.

There are three main decision problems regarding patterns and pattern languages. Those are the \emph{membership problem} (and its variations~\cite{gawrychowski_et_al:LIPIcs.MFCS.2021.48,Manea2022, Fleischmann2023}), the \emph{inclusion problem}, and the \emph{equivalence problem}. All are considered in the erasing (E) and non-erasing (NE) cases. The membership problem determines if a word belongs to the language of a pattern. This problem has been shown to be NP-complete for both, erasing- and non-erasing, pattern languages~\cite{DBLP:journals/jcss/Angluin80,Jiang1994}. The inclusion problem determines whether the language of one pattern is a subset of the language of another pattern. It has been shown to be generally undecidable by Jiang et al. in \cite{Jiang1995}. Freydenberger and Reidenbach~\cite{Freydenberger2010} as well as Bremer and Freydenberger~\cite{BREMER201215} improved that result and showed that it is undecidable for all bounded alphabets of size $|\Sigma| \geq 2$ for both erasing and non-erasing pattern languages. The equivalence problem asks whether the languages of two patterns are equal. For NE-pattern languages, this problem is trivially decidable and characterized by the equality of patterns up to a renaming of their variables~\cite{DBLP:journals/jcss/Angluin80}. The decidability of the erasing case is one of the major open problems in the field~\cite{Jiang1995,Reidenbach2004-1,Reidenbach2004-2,Ohlebusch1996,Reidenbach2007}. For terminal-free patterns, however, i.e., patterns without any terminal letters, the inclusion problem as well as the equivalence problem for E-pattern languages have been characterized and shown to be NP-complete~\cite{Jiang1995,DBLP:journals/ipl/EhrenfeuchtR79a}. In addition to that, in the case of terminal-free NE-pattern languages, Saarela~\cite{saarela:LIPIcs.ICALP.2020.140} has shown that the inclusion problem is undecidable in the case of a binary underlying alphabet.

Over time, various extensions to patterns and pattern languages have been introduced, either, to obtain additional expressibility due to some practical context or to get closer to an answer for the remaining open problems. Some examples are the bounded scope coincidence degree, patterns with bounded treewidth, $k$-local patterns, or strongly-nested patterns (see \cite{Day2018} and references therein). Koshiba~\cite{Koshiba1995} introduced so called \emph{typed patterns} that restrict substitutions of variables to types, i.e., arbitrary recursive languages. Geilke and Zilles~\cite{Geilke2011} extended this recently to the notion of \emph{relational patterns} and \emph{relational pattern languages}. In a similar more recent context and with specific relational constraints such as equal length, besides string equality, Freydenberger discusses in \cite{Freydenberger2018}, building on \cite{10.1145/2389241.2389250,FREYDENBERGER2013892,Freydenberger2017}, so called conjunctive regular path queries (CRPQs) with string equality and with length equality, which can be understood as systems of (relational) patterns with regular constraints (restricting variable substitutions to words of given regular languages) and with equal length constraints between variables.

In \cite{Nowotka2024}, a special form of typed patterns has been considered, i.e., patterns with regular constraints (also comparable to singleton sets of H-Systems \cite{FREYDENBERGER2013892}). Here, like mentioned before, variables may be restricted to arbitrary regular languages and the same variable may occur more than once. It has been shown that this notion suffices to obtain undecidability for both main open problems regarding pattern languages, i.e., the equivalence problem of E-pattern languages and the inclusion problem of terminal-free NE-pattern languages with an alphabet greater or equal to $3$. Another natural extension other than regular constraints is the notion of length constraints. Here, instead of restricting the choice of words for the substitution of variables, length constraints just restrict the lengths of substitution of variables in relation to each other. In the field of word equations, length constraints have been considered as a natural extension for a long time and, e.g., answering the decidability of the question whether word equations with length constraints have a solution, is a long outstanding problem (see, e.g., \cite{Lin2018QuadraticWE} and the references therein).

In this paper, we consider that natural extension of length constraints on patterns, resulting in the class of patterns called \emph{patterns with length constraints}. In general, we say that a \emph{length constraint} $\ell$ is a disjunction of systems of linear (diophantine) inequalities over the variables of $X$. We denote the \emph{set of all length constraints} by $\mathcal{C}_{Len}$. A \emph{pattern with length constraints} $(\alpha,\ell_\alpha)\in (\Sigma\cup X)^*\times\mathcal{C}_{Len}$ is a pattern associated with a length constraint. We say that a substitution $h$ is $\ell_\alpha$-valid if all variables are substituted according to $\ell_\alpha$. Now, the language of $(\alpha,\ell_\alpha)$ is defined analogously to pattern languages but restricted to $\ell_\alpha$-valid substitutions in the erasing- and non-erasing cases.

We examine erasing (E) and non-erasing (NE) pattern languages with length constraints. It can be shown, following existing results for patterns without additional constraints, that the membership problem for both cases in NP-complete. The inclusion problem is shown to be undecidable in both cases, too, notably even for terminal-free pattern languages, which is a difference to the decidability of the inclusion problem in the erasing case for pattern languages without any constraints, and which answers the case of alphabet sizes greater or equal to $3$ for non-erasing patterns without constraints. The main result of this paper is the undecidability of the equivalence problem for erasing pattern languages with length constraints in both cases, terminal-free and general, giving an answer to a problem of which the decidability has been an open problem for a long time in the case of no constraints. The final result shows that regular constraints and length constraints combined suffice to show undecidability of the equivalence problem for non-erasing pattern languages, a problem that is trivially decidable in case of no constraints and still open in the cases of just regular- or just length-constraints.

We begin by introducing the necessary notation in Section \ref{section:prelims} and then continue in Section \ref{section:patlang-len} with an examination of patterns with length constraints and their corresponding languages. Here, we will briefly discuss all results regarding the related decision problems (membership, inclusion, equivalence) while referring, due to extensiveness of the proofs, to the appendix for more details. In Section \ref{section:patlang-reglen}, we continue with an examination of patterns with regular and length constraints, also giving a full picture of the related decision problems. Finally, in Section \ref{section:conclusion}, we close this paper with a summary and discussion of the obtained results, the methods that were used, and the open problems that remain.

\section{Preliminaries}
\label{section:prelims}
Let $\N$ denote the natural numbers.
For $n,m \in \mathbb{N}$ set $[m,n] := \{k \in \mathbb{N} \mid m \leq k \leq n\}$. 
Denote $[n] := [1,n]$ and $[n]_0 := [0,n]$.
The powerset of any set $A$ is denoted by $\mathcal{P}(A)$. 
An \emph{alphabet} $\Sigma$ is a non-empty finite set whose elements are called \emph{letters}.  
A \emph{word} is a finite sequence of letters from $\Sigma$. 
Let $\Sigma^*$ be the set of all finite words over $\Sigma$, thus it is a free monoid with concatenation as operation and the empty 
word $\varepsilon$ as the neutral element. Set $\Sigma^+ := \Sigma^* \setminus \{\varepsilon\}$.
We call the number of letters in a word $w \in \Sigma^*$ \emph{length} of $w$, denoted by $|w|$.
Therefore, we have $|\varepsilon| = 0$.
Let $\Sigma^k$ denote the set of all words of length $k\in\N$ (resp. $\Sigma^{\leq k}$ or $\Sigma^{\geq k}$).
If $w = xyz$ for some $x,y,z\in\Sigma^*$, we call $x$ a \emph{prefix} of $w$, $y$ a \emph{factor} of $w$,
and $z$ a \emph{suffix} of $w$ and denote the sets of all prefixes, factors, and suffixes of $w$ by $\Pref(w)$, $\Fact(w)$,
and $\Suff(w)$ respectively.
For words $w,u\in\Sigma^*$, let $|w|_u$ denote the number of distinct occurrences of $u$ in $w$ as a factor.
Denote $\Sigma^k := \{w \in \Sigma^* \mid |w| = k\}$.
For $w \in \Sigma^*$, let $w[i]$ denote $w$'s $i^{th}$ letter for all $i \in [\vert w \vert]$. 
For reasons of compactness, we denote $w[i] \cdots w[j]$ by $w[i \cdots j]$ for all $i,j \in [\vert w \vert]$ with $i < j$.
Set $\al(w) :=  \{\ta \in \Sigma \mid \exists i \in [\vert w \vert] : w[i] = \ta\}$ as $w$’s alphabet.

Now, we introduce the notion of patterns and pattern languages with and without additional constraints (i.e., length constraints, length constraints, and a combination of both). After that, we briefly introduce the machine models used in the proofs of the main results of this paper.

\subsection{Patterns and Pattern Languages with Constraints}

Let $X$ be a countable set of variables such that $\Sigma \cap X = \emptyset$.
A \emph{pattern} is then a non-empty, finite word over $\Sigma \cup X$.
The set of all patterns over $\Sigma \cup X$ is denoted by $Pat_\Sigma$. 
For example, $x_1 \ta x_2 \tb \ta x_2 x_3$ is a pattern over $\Sigma = \{\ta,\tb\}$ with $x_1,x_2,x_3\in X$.
For a pattern $\alpha\in Pat_\Sigma$, let $\var(\alpha) := \{\ x \in X\ |\ |\alpha|_x \geq 1\ \}$ denote the set of variables occurring in $p$.
A \emph{substitution of $\alpha$} is a morphism $h : (\Sigma \cup X)^* \to \Sigma^*$ such that $h(\ta) = \ta$ for all $\ta \in \Sigma$ and 
$h(x) \in \Sigma^*$ for all $x \in X$. 
If we have $h(x) \neq \varepsilon$ for all $x \in \var(\alpha)$, we call $h$ a \emph{non-erasing substitution} for $\alpha$. 
Otherwise, $h$ is an \emph{erasing substitution} for $\alpha$. The set of all substitutions w.r.t.~$\Sigma$ is denoted by $H_\Sigma$.
If $\Sigma$ is clear from the context, we may write just $H$.
Given a pattern $\alpha\in Pat_\Sigma$, its erasing pattern language $L_E(\alpha)$ and its non-erasing pattern language $L_{NE}(\alpha)$
are defined respectively by
\begin{align*}
	L_{E}(\alpha) &:= \{\ h(\alpha)\ |\ h\in H, h(x) \in \Sigma^* \text{ for all } x\in\var(\alpha)\}, \text{ and } \\
	L_{NE}(\alpha) &:= \{\ h(\alpha)\ |\ h\in H, h(x) \in \Sigma^+ \text{ for all } x\in\var(\alpha)\}.
\end{align*}

A \emph{length constraint} $\ell$ is a disjunction of systems of linear diophantine inequalities over variables of $X$.
We denote the \emph{set of all length constraints} by $\mathcal{C}_{Len}$. 
A \emph{pattern with length constraints} is a pair $(\alpha,\ell_\alpha)\in\Pat_\Sigma\times\mathcal{C}_{Len}$ where all variables occurring in $\ell_\alpha$ must occur in $\alpha$.
We denote the \emph{set of all patterns with length constraints} by $\PatLen$.
For some $(\alpha,\ell_\alpha)\in\PatLen$ and $h\in H$, we say that $h$ is a \emph{$\ell_\alpha$-valid substitution} if $\ell_\alpha$ is satisfied when associating each variable $x\in\var(\alpha)$ in $\ell_\alpha$ with the value $|h(x)|$, i.e., the length of the substitution of the variable $x$. Consider the following example.
\begin{example}
    Let $\Sigma = \{\ta,\tb\}$ and $\alpha = x_1\ \ta\ x_2\ \ta\ x_1$. Assume we have the length constraint $\ell_\alpha$ defined by the following system of linear diophantine inequalities:
    \begin{align*}
        2x_1 + x_2 \leq 5\\
        x_2 \geq 1
    \end{align*}
    Then, we know that any $\ell_\alpha$-valid substitution $h\in H$ cannot have $h(x_1) = u$ for some $u\in\Sigma^*$ with $|u| \geq 3$, as this would already imply $2|u| = 2|h(x)| \geq 2\cdot3 = 6$ and $6 \geq 5$. Also, we see that $h(x_2)\neq\varepsilon$ as the second constraint demands a substitution of length at least 1. So, for example we could have $h(\alpha) = \tb\tb\ta\tb\ta\tb\tb$ or $h(\alpha) = \tb\ta\tb\ta\tb\ta\tb$ but not $h(\alpha) = \ta\ta$ or $h(\alpha) = \tb\tb\ta\ta\ta\ta\tb\tb$. 
\end{example}
We denote the set of all $\ell_\alpha$-valid substitutions by $H_{\ell_{\alpha}}$.
The notion of pattern languages is extended by the following. For any $(\alpha,\ell_\alpha)\in\PatLen$ we denote by 
    $$L_E(\alpha,\ell_\alpha) := \{\ h(\alpha)\ |\ h\in H_{\ell_{\alpha}}, h(x)\in\Sigma^*\text{ for all } x\in\var(\alpha)\ \}$$
the \emph{erasing pattern language with length constraints} of $(\alpha,\ell_\alpha)$ and by
    $$L_{NE}(\alpha,\ell_\alpha) := \{\ h(\alpha)\ |\ h\in H_{\ell_{\alpha}}, h(x)\in\Sigma^+\text{ for all } x\in\var(p)\ \}$$
the \emph{non-erasing pattern language with length constraints} of $(\alpha,\ell_\alpha)$.

Similar to length constraints, we can define \emph{regular constraints} for variables in a pattern.
Let $\mathcal{L}_{Reg}$ be the set of all regular languages. We call a mapping $r : X \rightarrow \mathcal{L}_{Reg}$
a \emph{regular constraint} on $X$. If not stated otherwise, we always have $r(x) = \Sigma^*$.
We denote the \emph{set of all regular constraints} by $\mathcal{C}_{Reg}$.
For some $r\in\mathcal{C}_{Reg}$ we define the \emph{language of a variable} $x\in X$ by $L_r(x) = r(x)$.
If $r$ is clear by the context, we omit it and just write $L(x)$.
A \emph{pattern with regular constraints} is a pair $(\alpha,r_\alpha) \in Pat_\Sigma\times\mathcal{C}_{Reg}$.
We denote the \emph{set of all patterns with regular constraints} by $\PatReg$.
For some $(\alpha,r_\alpha)\in \PatReg$ and $h\in H$, we say that $h$ is a \emph{$r_\alpha$-valid substitution}
if $h(x) \in L(x)$ for all $x\in\var(\alpha)$. The set of all $r_\alpha$-valid substitutions is denoted by $H_{r_p}$.
Given some $(\alpha,r_\alpha)\in\PatReg$, we analogously define the \emph{erasing- and non-erasing pattern languages with regular constraints} $L_E(\alpha,r_\alpha)$ and $L_{NE}(\alpha,r_\alpha)$ over $H_{r_\alpha}$ as we did for length constraints.

Combining both, we say that a triple $(\alpha,r_\alpha,\ell_\alpha)\in\Pat\times\mathcal{C}_{Reg}\times\mathcal{C}_{Len}$ is a \emph{pattern with regular and length constraints} and denote the \emph{set of all patterns with regular and length constraints} by $\PatRegLen$. Given some $(\alpha,r_\alpha,\ell_\alpha)\in\PatRegLen$, we say that a substitution $h\in H$ is \emph{$r_\alpha$-$\ell_\alpha$-valid} if it is $r_\alpha$-valid and $\ell_\alpha$-valid and denote the set of all $r_\alpha$-$\ell_\alpha$-valid substitutions by $H_{r_\alpha,\ell_\alpha}$. Additionally, we analogously define the \emph{erasing- and non-erasing pattern languages with regular and length constraints} $L_E(\alpha,r_\alpha,\ell_\alpha)$ and $L_{NE}(\alpha,r_\alpha,\ell_\alpha)$ over $H_{r_\alpha,\ell_\alpha}$ as we did in the previous two cases for length constraints and regular constraints.

In the proofs of our main results Theorem \ref{theorem:patlang-len-erasing-equivalence-undecidable}, Theorem \ref{theorem:patlang-len-nonerasing-inclusion-tf-undecidability}, and Theorem \ref{theorem:patlang-reglen-nonerasing-equivalence-undecidability}, we use two different automata models with undecidable emptiness problems to obtain each result. We will use the notion of \emph{nondeterministic 2-counter automata without input} (see, e.g., \cite{Iberra1978}), as well as, the notion of a very specific universal Turing machine $U$ as it is used in \cite{BREMER201215}. As mentioned in the introduction, due to their extensive lengths, the proofs of each of the main theorems is found in the appendix. In the main body, only a rough explanation, not relying on a formal definition of the machine types, of each proof idea is provided. Hence, the formal definition of the notion of \emph{nondeterministic 2-counter automata without input} is found in Appendix~\ref{section:defAut}. As the construction and the arguments in the proof of Theorem \ref{theorem:patlang-len-nonerasing-inclusion-tf-undecidability} work independently from the formal definition of $U$, we refer to \cite{BREMER201215} for more details. Regarding both machines, assuming $A$ is a machine of one of the previously mentioned types, in the following, $\ValC(A)$ denotes the set of all binary encodings (with a specific format) of valid computation of $A$.

We continue with our overview of the results regarding pattern languages with length constraints.

\section{Results for Pattern Languages with Length Constraints}
\label{section:patlang-len}
To begin developing an understanding of the additional expressiveness gained by allowing for length constraints, notice the following observation for patterns with length or with regular constraints which does not necessarily hold for patterns without any constraints.

\begin{lemma}
    For each pattern with length constraints $(\alpha,\ell_\alpha)\in\PatLen$ (and for each pattern with regular constraints $(\beta,r_\beta)\in \PatReg$), there exists some adapted set of length constraints $\ell_\alpha'\in\mathcal{C}_{Len}$ (resp., some adapted set of regular constraints $r_\beta'\in\mathcal{C}_{Reg}$) such that $L_{NE}(\alpha,\ell_\alpha) = L_{E}(\alpha,\ell_\alpha')$ (and $L_{NE}(\beta,r_\beta) = L_E(\beta,r_\beta)$ resp.).
\end{lemma}
\begin{proof}
    Indeed, given some pattern with length constraints $(\alpha,\ell_\alpha)\in\PatLen$, we can define the length constraint $\ell_\alpha'$ by using all constraints in $\ell_\alpha$ and additionally, for each $x\in\var(\alpha)$, adding the constraint $x \geq 1$ to $\ell_\alpha'$. Then, $L_{E}(\alpha,\ell_\alpha') = L_{NE}(\alpha,\ell_\alpha)$.
    
    We obtain the same result for pattern languages with regular constraints by intersecting the language of all variables with $\Sigma^+$, i.e., given any language $L(x)$ for some variable $x\in X$ that is defined by some regular constraint $r\in\mathcal{C}_{Reg}$, we can define a regular constraint $r'$ that defines the language $L'(x) = L(x) \cap \Sigma^+$. Then, given some pattern $\alpha\in\Pat$, we obtain $L_{NE}(\alpha,r) = L_E(\alpha,r')$.
    \end{proof}

The following statement then immediately follows by the previous lemma.

\begin{corollary}\label{corollary:patlang-len-erasing-atleastashardas-nonerasing}
    Solving any problem for erasing pattern languages with length (or regular) constraints is at least as hard as solving the same problem for non-erasing pattern languages with length constraints.
\end{corollary}

As shown in \cite{Nowotka2024}, considering regular constraints, problem cases over all patterns can be reduced to problems involving terminal-free patterns. The same does not hold in general in the case of length constraints, witnessed by the following proposition.

\begin{proposition}
    There exists $(\alpha,\ell_\alpha)\in\PatLen$ such that no $(\beta,\ell_\beta)\in\PatLen$ with a terminal-free pattern $\beta\in X^*$ exists, for which we have that $L_X(\alpha,\ell_\alpha) = L_X(\beta,\ell_\beta)$ for $X\in\{E,NE\}$.
\end{proposition}
\begin{proof}
    Assume $|\Sigma| \geq 2$ and assume w.l.o.g. $\ta,\tb\in\Sigma$ with $\ta \neq \tb$. Let $(\alpha,\ell_\alpha)\in\PatLen$ such that $\alpha = x_1\ta$ and $\ell_\alpha$ is an empty length constraint. So, $L_E(\alpha,\ell_\alpha) = L_E(\alpha) = \Sigma^*\cdot\{\ta\}$ ($L_{NE}(\alpha,\ell_\alpha) = L_{NE}(\alpha) = \Sigma^+\cdot\{\ta\}$). Suppose there exists some $\beta\in X^*$ and length constraint $\ell_\beta$ such that $L_E(\beta,\ell_\beta) = L_E(\alpha,\ell_\alpha)$ ($L_{NE}(\beta,\ell_\beta) = L_{NE}(\alpha,\ell_\alpha)$). W.l.o.g. we continue with the erasing case. The following also holds for the non-erasing case with small changes. Let $w\in L_E(\alpha,\ell_\alpha)$. Then $w = u\ta$ for some $u\in\Sigma^*$. As $L_E(\alpha,\ell_\alpha) = L_E(\beta,\ell_\beta)$, there exists some $h\in H_{\ell_\beta}$ such that $h(\beta) = w = u\ta$. Then, $\beta = \beta_1 x \beta_2$ for $x\in X$ and $\beta_1,\beta_2\in X^*$ as well as $u = u_1u_2$ for some $u_1,u_2\in\Sigma^*$ such that $h(\beta_1) = u_1$, $h(x) = u_2\ta$ and $h(\beta_2) = \varepsilon$. But then, we also find some $h'\in H_{\ell_\beta}$ with $h(x) = u_1\tb$ as $\ta$ and $\tb$ are obtained at the same position by the same variable.
    Then, $h'(\beta) = v\tb$ for some $v\in\Sigma^*$. As $v\tb \notin \Sigma^* \cdot \{\ta\}$ but $L_E(\alpha,\ell_\alpha) = \Sigma^* \cdot \{\ta\}$, we know that $\beta$ cannot exist.
\end{proof}

With that, we know that we cannot restrict ourselves to only show hardness or undecidability for the general case, as the terminal free case may result in something else. We begin with the membership problems for both, the erasing- and non-erasing pattern languages. Those have been shown to be NP-complete for patterns without constraints in the terminal-free and general cases (see, e.g.,\cite{DBLP:journals/jcss/Angluin80,Jiang1994,Schmid2012}). Hence, we observe the following for patterns with length constraints.

\begin{corollary}\label{corollary:patlan-len-membership-npc}
	Let $(\alpha,\ell_\alpha) \in \PatLen$ and $w\in\Sigma^*$.
	The decision problem of whether $w \in L_X(\alpha,\ell_\alpha)$ for $X \in \{E,NE\}$ is NP-complete, even if the considered pattern $\alpha$ is terminal-free.
\end{corollary}

Indeed, NP-hardness is immediately obtained by the NP-hardness of patterns without any constraints. Given some $(\alpha,\ell_\alpha)\in\PatLen$, NP containment follows by the fact that any valid certificate results in a substitution $h\in H_{\ell_\alpha}$ of length at most $|h(\alpha)| = |w|$ and that we can check in polynomial time whether some given substitution is $\ell_\alpha$-valid. This can be done by checking whether the lengths of the variable substitutions $|h(x)|$ for all $x\in\var(\alpha)$ satisfy $\ell_\alpha$. This concludes the membership problem for all cases.

Another result that can be immediately obtained is the general undecidability of the inclusion problem for pattern languages with length constraints. For pattern languages without any constraints, this problem has been generally shown to be undecidable, first, for unbounded alphabets by Jiang et al. \cite{Jiang1995} and later on for all bounded alphabets of sizes $|\Sigma| \geq 2$ by Freydenberger and Reidenbach in \cite{Freydenberger2010} as well as Bremer and Freydenberger in \cite{BREMER201215}. So, we have the following.

\begin{theorem}\label{theorem:pattern-inclusion-undecidability}\cite{Jiang1995, Freydenberger2010, BREMER201215}
	Let $\alpha,\beta\in Pat_\Sigma$.
	In general, for all alphabets $\Sigma$ with $|\Sigma| \geq 2$, it is undecidable to answer whether $L_X(\alpha) \subseteq L_X(\beta)$ for $X\in\{E,NE\}$.
\end{theorem}

Out of that, in the general case, we immediately obtain the following for pattern languages with length constraints.

\begin{corollary}
    Let $(\alpha,\ell_\alpha),(\beta,\ell_\beta)\in\PatLen$. For all alphabets $\Sigma$ with $|\Sigma|\geq2$ it is undecidable to answer whether $L_X(\alpha,\ell_\alpha) \subseteq L_X(\beta,\ell_\beta)$ for $X\in\{E,NE\}$.
\end{corollary}

That leaves us with the terminal free cases of the inclusion problem as well as the general and terminal-free cases of the equivalence problems for both, erasing and non-erasing pattern languages with length constraints. The first and most significant main result shows that the most prominent open problem for pattern languages, i.e., the equivalence problem for erasing patterns, is undecidable for pattern languages with length constraints, even if we are restricted to terminal-free patterns and use no disjunctions in the length constraints.

\begin{theorem}\label{theorem:patlang-len-erasing-equivalence-undecidable}
    Let $(\alpha,\ell_\alpha),(\beta,\ell_\beta)\in\PatLen$. The problem of deciding whether we have $L_E(\alpha,\ell_\alpha) = L_E(\beta,\ell_\beta)$ is undecidable for all alphabets $\Sigma$ with $|\Sigma| \geq 2$, even if $\alpha$ and $\beta$ are restricted to be terminal-free and even if $\ell_\alpha$ and $\ell_\beta$ have no disjunctions.
\end{theorem}

The proof of Theorem \ref{theorem:patlang-len-erasing-equivalence-undecidable}, due to its extent, can be found in Appendix \ref{section:proof1}. Here, we give a rough sketch of the idea. Based on the proofs in, e.g., \cite{Jiang1995} and \cite{Freydenberger2010}, we reduce the emptiness problem for nondeterministic 2-counter automaton without input to the equivalence problem of erasing pattern languages with length constraints. Given such an automaton $A$, we construct two patterns with length constraints $(\alpha,\ell_\alpha)$ and $(\beta,\ell_\beta)$ such that $L_E(\alpha,\ell_\alpha) = L_E(\beta,\ell_\beta)$ if and only if $\ValC(A) = \emptyset$. We mention, here, that the constructed length constraints $\ell_\alpha$ and $\ell_\beta$ respectively consist only of a single system of linear inequalities, hence, they use no disjunctions and are, by that, rather simple. The first pattern $(\alpha,\ell_\alpha)$ can be understood as a pattern that allows for arbitrary substitutions, restricted to some structural boundaries. The pattern $(\beta,\ell_\beta)$ can be understood as a \emph{tester} pattern that has a certain number of sub-patterns, called \emph{predicates}. Out of these, one has to always be picked to obtain some substitution with a specific structure similar to the one given for the words obtained from $(\alpha,\ell_\alpha)$. First, the proof shows that using the construction given in Appendix \ref{section:proof1}, we have that all words in $L_E(\beta,\ell_\beta)$ are actually also in in $L_E(\alpha,\ell_\alpha)$, i.e., $L_E(\beta,\ell_\beta) \subseteq L_E(\alpha,\ell_\alpha)$. This is a significant difference to the constructions established in the undecidability proofs for the inclusion problem for pattern languages in \cite{Jiang1995,Freydenberger2010} and the main reason for this idea to actually work for the equivalence problem instead of only the inclusion problem. This is similar to the idea of the construction in \cite{Nowotka2024} for patterns with regular constraints, where the same approach but with differently constructed patterns and with another type of constraints has been used. Similar to the constructions given in \cite{Jiang1995,Freydenberger2010}, it is shown that a word in $L_E(\alpha,\ell_\alpha)$, with specific structure, always needs the use of a predicate in $(\beta,\ell_\beta)$ to obtain that word in $L_E(\beta,\ell_\beta)$. Using that, predicates can be constructed that allow for all words with the given specific structure and without an encoding of a valid computation of $A$ as a factor to be obtained in $L_E(\beta,\ell_\beta)$. This leaves only those words with specific structure and with encodings of valid computations from $A$ as words that only occur in $L_E(\alpha,\ell_\alpha)$. Hence, if no valid computations exist, i.e., $\ValC(A) = \emptyset$, then all words in $L_E(\alpha,\ell_\alpha)$ are also in $L_E(\beta,\ell_\beta)$, as there always exists a predicate in $(\beta,\ell_\beta)$ to obtain these words. Together with the property that $L_E(\beta,\ell_\beta) \subseteq L_E(\alpha,\ell_\alpha)$, we have that the equivalence problem, i.e., the property whether $L_E(\alpha,\ell_\alpha) = L_E(\beta,\ell_\beta)$, decides whether $\ValC(A)$ is empty or not, hence deciding the emptiness problem of nondeterministic 2-counter automata without input. The main body of the proof shows the case for an alphabet $\Sigma$ of size $|\Sigma| = 2$. An explanation of how to extend this construction to arbitrary alphabets is provided in the end of Appendix \ref{section:proof1} as well.

Due to Theorem \ref{theorem:patlang-len-erasing-equivalence-undecidable}, the undecidability of the inclusion problem for terminal-free erasing pattern languages immediately follows. 

\begin{corollary}
    Let $(\alpha,\ell_\alpha),(\beta,\ell_\beta)\in\PatLen$ for terminal-free patterns $\alpha,\beta\in X^*$. The problem of answering whether we have $L_E(\alpha,\ell_\alpha) \subseteq L_E(\beta,\ell_\beta)$ is undecidable for all alphabets $\Sigma$ with $|\Sigma| \geq 2$.
\end{corollary}
\begin{proof}
    Suppose it was decidable. Then, we could decide whether $L_E(\alpha,\ell_\alpha) \subseteq L_E(\beta,\ell_\beta)$ and whether $L_E(\beta,\ell_\beta) \subseteq L_E(\alpha,\ell_\alpha)$ and by that decide if we have $L_E(\alpha,\ell_\alpha) = L_E(\beta,\ell_\beta)$. That is a contradiction to Theorem \ref{theorem:patlang-len-erasing-equivalence-undecidable}.
\end{proof}

The second main result of the paper, and the last one regarding patterns with only length constraints, shows that the inclusion problem for terminal-free non-erasing pattern languages with length constraints is also undecidable for all alphabets $\Sigma$ with $|\Sigma| > 2$. In \cite{saarela:LIPIcs.ICALP.2020.140}, Saarela has shown this problem to be undecidable for patterns without any constraints in the binary case.

\begin{theorem}\cite{saarela:LIPIcs.ICALP.2020.140}
    Let $\alpha,\beta\in X^*$ be two terminal-free patterns. Answering whether $L_{NE}(\alpha) \subseteq L_{NE}(\beta)$ is undecidable for alphabets $\Sigma$ with $|\Sigma| = 2$.
\end{theorem}

Hence, the undecidability of the inclusion problem for pattern languages with length constraints follows immediately in this case.

\begin{corollary}\label{cor:result-incl-tf-binary-undec-src}
    Let $(\alpha,\ell_\alpha),(\beta,\ell_\beta)\in\PatLen$ for terminal-free patterns $\alpha,\beta\in X^*$. Answering whether we have $L_{NE}(\alpha,\ell_\alpha) \subseteq L_{NE}(\beta,\ell_\beta)$ is undecidable for alphabets $\Sigma$ with $|\Sigma| = 2$.
\end{corollary}

In the case of length constraints, a general extension to all alphabet sizes is possible. We obtain the following.

\begin{theorem}\label{theorem:patlang-len-nonerasing-inclusion-tf-undecidability}
    Let $(\alpha,\ell_\alpha),(\beta,\ell_\beta)\in\PatLen$ for terminal-free patterns $\alpha,\beta\in X^*$. Answering whether we have $L_{NE}(\alpha,\ell_\alpha) \subseteq L_{NE}(\beta,\ell_\beta)$ is undecidable for alphabets $\Sigma$ with $|\Sigma| \geq 2$, even if $\ell_\alpha$ and $\ell_\beta$ have no disjunctions.
\end{theorem}

The proof of Theorem \ref{theorem:patlang-len-nonerasing-inclusion-tf-undecidability} can be found in Appendix \ref{section:proof2}. Again, here, we only give a rough sketch of the idea. The binary case follows by Corollary \ref{cor:result-incl-tf-binary-undec-src}, i.e., the result by Saarela \cite{saarela:LIPIcs.ICALP.2020.140}. The general proof idea is based on the construction done by Freydenberger and Bremer \cite{BREMER201215} for the binary case and significantly adapted in its extension to arbitrary alphabets. As the extension to larger alphabets is strongly based on the construction for the binary case, a complete construction is given for both cases to aid understandability. We reduce the emptiness problem of a specific universal Turing Machine $U$ (the one mentioned in the preliminaries) to the inclusion problem of non-erasing pattern languages with length constraints for terminal-free patterns. The idea is relatively similar to the idea for the proof of Theorem \ref{theorem:patlang-len-erasing-equivalence-undecidable} (or the proof in \cite{BREMER201215}) in that is utilizes a pattern $(\alpha,\ell_\alpha)$ in which all words with a specific structure can be obtained, and a pattern $(\beta,\ell_\beta)$ that requires the use of predicates to obtain words from $L_{NE}(\alpha,\ell_\alpha)$. Again, we mention that the constructed length constraints $\ell_\alpha$ and $\ell_\beta$ respectively consist only of a single system of linear inequalities, thus, they use no disjunctions and are, by that, rather simple. As this result only shows undecidability of the inclusion problem, we do not have the property $L_{NE}(\beta,\ell_\beta)\subseteq L_{NE}(\alpha,\ell_\alpha)$ and only use the case distinction whether $L_{NE}(\alpha,\ell_\alpha) \subseteq L_{NE}(\beta,\ell_\beta)$ to decide the emptiness of $\ValC(U)$. For the binary case, we essentially take the constructed patterns from \cite{BREMER201215} but swap out any occurrence of the letter $0$ with a new variable $x_0$ and any occurrence of the letter $x_\#$ with a new variable $x_\#$ (also similar, to some extend, to the idea in \cite{saarela:LIPIcs.ICALP.2020.140}). In that sense, the patterns look very similar to the constructed patterns in \cite{BREMER201215}, but do not contain any terminals. The addition of length constraints allows for the length of any of the substitutions for $x_0$ and $x_\#$ to be bounded by length $1$. Adaptations to the constructed patterns allow to show a property that any word obtained from $(\alpha,\ell_\alpha)$ can be obtained from $(\beta,\ell_\beta)$ if $x_0$ and $x_\#$ are substituted by the same letter. If $x_0$ and $x_\#$ are substituted by different letters, it is shown that essentially the same logic as in the proof done in \cite{BREMER201215} can be applied here, and by that, the binary case is concluded. Now, assume $\Sigma = \{0,\#,\ta_1,...,\ta_\sigma\}$ for some $\sigma > 0$. For each new letter $\ta_i$ with $i\in[\sigma]$, a new variable $x_{\ta_i}$ is introduced and bounded in its length by $1$ by using the length constraints in both patterns. Then, factors consisting of only these variables are prepended to the patterns $\alpha$ and $\beta$ as well as in certain positions in the middle parts of both patterns. This allows for an extension of the property that if any two of the variables $x_0$, $x_\#$, and $x_{\ta_1}$ to $x_{\ta_\sigma}$ are substituted by equal letters in a substitution for $(\alpha,\ell_\alpha)$, then the resulting word is always contained in $L_{NE}(\beta,\ell_\beta)$. In addition to that, predicates are added that can be used if all variables $x_0$, $x_\#$, and $x_{\ta_1}$ to $x_{\ta_\sigma}$ are substituted by different letters, but any of the letters obtained from the variables $x_{\ta_1}$ to $x_{\ta_\sigma}$ appear later on in the part that may contain the encodings of valid computations from $U$. By that, using these and the previous predicates, only if a substitution on $(\alpha,\ell_\alpha)$ replaces $x_0$, $x_\#$, and $x_{\ta_1}$ to $x_{\ta_\sigma}$ with pairwise distinct letters and a specific middle part in $\alpha$ is substituted by the encoding of a valid computation from $U$-using the letters obtained by the substitutions of $x_0$ and $x_\#$-we cannot obtain that word from $(\beta,\ell_\beta)$. Hence, if and only if $\ValC(U)$ has no valid computations, we have that $L_{NE}(\alpha,\ell_\alpha)\subseteq L_{NE}(\beta,\ell_\beta)$.

\begin{remark}
    Indeed, by Corollary \ref{corollary:patlang-len-erasing-atleastashardas-nonerasing}, Theorem \ref{theorem:patlang-len-nonerasing-inclusion-tf-undecidability} also implies the undecidability of the inclusion problem of terminal-free erasing pattern languages with length constraints. Hence, this is an additional approach to obtain that result other than Theorem \ref{theorem:patlang-len-erasing-equivalence-undecidable}.
\end{remark}

Whether the equivalence problem of non-erasing pattern languages with length constraints is decidable, is left as an open question with no specific conjecture so far. Interestingly, also in the case of regular constraints, no definite answer for the decidability of this problem could be given, yet (cf. \cite{Nowotka2024}). A summary of the current results can be found in Table \ref{tab:complexity-full} at the end of the final discussion.

\label{section:patlang-len}

\section{Result for Pattern Languages with Regular and Length Constraints}
\label{section:patlang-reglen}
As there are still open problems for pattern languages with regular (or length constraints), i.e., the equivalence problem for non-erasing pattern languages with regular (or length-constraints), finding an answer to this problem for a larger class of languages using both constraint-types is well motivated. Indeed, considering pattern languages with regular and length constraints suffices to obtain the undecidability of that problem. That is noticeable, since this problem is trivially decidable for patterns without any constraints. We begin with mentioning all results for decision problems on pattern languages with regular and length constraints that we can immediately obtain from pervious results. 

For the membership, we immediately obtain NP-completeness for all cases, i.e., erasing- and non-erasing as well as terminal-free and general.

\begin{corollary}
    Let $(\alpha,r_\alpha,\ell_\alpha)\in\PatRegLen$ and $w\in\Sigma^*$. The decision problem of whether $w\in L_{X}(\alpha,r_\alpha,\ell_\alpha)$ for $X\in\{E,NE\}$ is NP-complete, even if the considered pattern $\alpha$ is terminal-free.
\end{corollary}

Indeed, NP-hardness follows, as before, from the NP-hardness of pattern languages without any constraints. NP-containment follows from the fact NP-containment can be checked in the cases of pattern languages with regular-constraint as well as pattern languages with length-constraints and that no additional properties have to be checked.

As this class of pattern languages utilizes regular constraints, we immediately obtain the property mentioned before, and shown in \cite{Nowotka2024}, that we can always find a terminal-free pattern with regular constraints that expresses the same language as a general pattern with regular constraints. 

\begin{proposition}\label{proposition:patlan-reg-tf-same-as-general}\cite{Nowotka2024}
    Let $(\alpha,r_\alpha)\in\PatReg$ be some pattern with regular constraints. Then, there exists some terminal free pattern $\beta\in X^*$ and a regular constraint $r_\beta$ such that $L_X(\alpha,r_\alpha) = L_X(\beta,r_\beta)$ for $X\in\{E,NE\}$.
\end{proposition}

Using the same logic as used to show that property, we immediately obtain the same for pattern languages with regular and length constraints, as, given some pattern with regular and length constraints $(\alpha,r_\alpha,\ell_\alpha)\in\PatRegLen$, one can introduce for each terminal letter $\mathtt{a}$ a new variable $x_\mathtt{a}$ and set $L(x_{\mathtt{a}}) = \{\mathtt{a}\}$ while not changing any length constraints.

\begin{corollary}\label{corollary:patlan-reglen-tf-same-as-general}
    Let $(\alpha,r_\alpha,\ell_\alpha)\in\PatRegLen$ be some pattern with regular- and length constraints. Then, there exists some terminal free pattern $\beta\in X^*$, a regular constraint $r_\beta$, and a length constraint $\ell_\beta$ such that $L_X(\alpha,r_\alpha,\ell_\beta) = L_X(\beta,r_\beta,\ell_\beta)$ for $X\in\{E,NE\}$.
\end{corollary}

Additionally, as the inclusion problem is undecidable for all cases (E, NE, terminal-free, general) for pattern languages with regular or with length constraints, the same follows immediately for pattern languages with regular and length constraints.

\begin{corollary}
    Let $(\alpha,r_\alpha,\ell_\alpha),(\beta,r_\beta,\ell_\beta)\in\PatRegLen$. It is undecidable to answer whether $L_X(\alpha,r_\alpha,\ell_\alpha)\subseteq L_X(\beta,r_\beta,\ell_\beta)$ for $X\in\{E,NE\}$ for all alphabets $\Sigma$ with $|\Sigma|\geq 2$, even if $\alpha$ and $\beta$ are restricted to be terminal-free.
\end{corollary}

The same follows for the equivalence problem of general and terminal-free erasing pattern languages with regular- and length-constraints, as also here we have the undecidability for both cases, regular constraints and length constraints, respectively. We obtain the following.

\begin{corollary}
    Let $(\alpha,r_\alpha,\ell_\alpha),(\beta,r_\beta,\ell_\beta)\in\PatRegLen$. It is undecidable to answer whether $L_E(\alpha,r_\alpha,\ell_\alpha)= L_E(\beta,r_\beta,\ell_\beta)$ for all alphabets $\Sigma$ with $|\Sigma|\geq 2$, even if $\alpha$ and $\beta$ are restricted to be terminal-free.
\end{corollary}

For all results obtained so far, no disjunction of systems of linear inequalities have had to be applied. For the equivalence problem of non-erasing pattern languages with regular and length constraints, making use of disjunctions in length constraints in addition to regular constraints allows us to obtain undecidability in this case. The following theorem concludes the third and final main result of this paper.

\begin{theorem}\label{theorem:patlang-reglen-nonerasing-equivalence-undecidability}
    Let $(\alpha,r_\alpha,\ell_\alpha),(\beta,r_\beta,\ell_\beta)\in\PatRegLen$. It is undecidable to answer whether $L_{NE}(\alpha,r_\alpha,\ell_\alpha) = L_{NE}(\beta,r_\beta,\ell_\beta)$ for all alphabets $\Sigma$ with $|\Sigma| \geq 2$, even if $\alpha$ and $\beta$ are restricted to be terminal-free.
\end{theorem}

Due to its extensive length, the proof of Theorem \ref{theorem:patlang-reglen-nonerasing-equivalence-undecidability} can be found in Appendix \ref{section:proof3}. Its logic works similar to the proof of Theorem \ref{theorem:patlang-len-erasing-equivalence-undecidable} and the proof of the main result in \cite{Nowotka2024}, but needs significant adaptations to work for non-erasing pattern languages. Again, here, we only give a rough sketch of the idea of the proof. Similar to the proof of Theorem \ref{theorem:patlang-len-erasing-equivalence-undecidable}, we reduce the emptiness problem for nondeterministic 2-counter automata to the equivalence problem of non-erasing pattern languages with regular and length constraints. So, given some nondeterministic 2-counter automaton without input $A$, we construct two patterns with regular and length constraints $(\alpha,r_\alpha,\ell_\alpha)$ and $(\beta,r_\beta,\ell_\beta)$ such that $L_{NE}(\alpha,r_\alpha,\ell_\alpha) = L_{NE}(\beta,r_\beta,\ell_\beta)$ if and only if $\ValC(A) = \emptyset$. In contrast to the proofs of Theorem \ref{theorem:patlang-len-erasing-equivalence-undecidable} and Theorem \ref{theorem:patlang-len-nonerasing-inclusion-tf-undecidability}, this time, the length constraints $\ell_\beta$ use a disjunction of systems of linear diophantine inequalities while $(\alpha,r_\alpha,\ell_\alpha)$ has no regular or length constraints at all. Again, $(\alpha,r_\alpha,\ell_\alpha)$ serves as a pattern to obtain any word with a specific structure and $(\beta,r_\beta,\ell_\beta)$ serves as a pattern to obtain words with that specific structure that fulfill given properties. These properties are, again, defined by sub-patterns in $(\beta,r_\beta,\ell_\beta)$ that are called predicates. The regular and length constraints $r_\beta$ and $\ell_\beta$ can be set in such a way that always a single predicate has to be chosen, while the variables outside of these predicates have to be substituted in a very fixed way. This results in words that have equal prefixes and suffixes in comparison to all words obtained from $(\alpha,r_\alpha,\ell_\alpha)$, while the middle part is then checked with the predicates. These predicates are constructed similar to the ones in \cite{Nowotka2024}, just with slight adaptations, as now we have to use non-erasing substitutions instead of erasing ones. In the end, it is shown, that the predicates cover all words where the middle part is not an encoding of a valid computation in $\ValC(A)$. Hence, if $\ValC(A) = \emptyset$, we have $L_{NE}(\alpha,r_\alpha,\ell_\alpha) \subseteq L_{NE}(\beta,r_\beta,\ell_\beta)$. Additionally, by the way these patterns with regular and length constraints are constructed, it is shown that we have $ L_{NE}(\beta,r_\beta,\ell_\beta) \subseteq L_{NE}(\alpha,r_\alpha,\ell_\alpha)$, similar to the idea of the proof of Theorem \ref{theorem:patlang-len-erasing-equivalence-undecidable}. Hence, we have $L_{NE}(\alpha,r_\alpha,\ell_\alpha) = L_{NE}(\eta,r_\beta,\ell_\beta)$ if and only if $\ValC(A) = \emptyset$. The main body shows this result for binary alphabets $\Sigma$ with $|\Sigma| = 2$. To extend this to arbitrary alphabet sizes $\Sigma$ with $|\Sigma| \geq 3$, a similar approach as in \cite{Nowotka2024} can be applied. It can be shown that the binary case of the equivalence problem for non-erasing pattern languages with regular- and length constraints can be reduced to the same problem over all other alphabet sizes that are larger. As the regular constraints $r_\alpha$ and $r_\beta$ allow for the effective restriction of the alphabets used of each variable, no matter which alphabet size is used for the patterns, the predicate checks can essentially be forced to be done over binary words. 

The construction of the proof for Theorem \ref{theorem:patlang-reglen-nonerasing-equivalence-undecidability} uses patterns that consist partly of terminal-symbols. Using Corollary \ref{corollary:patlan-reglen-tf-same-as-general}, however, we know that we can always find terminal-free patterns that produce the exact same languages. Hence, the result also follows immediately for terminal-free patterns.

This concludes all results regarding decision problems for pattern languages with regular- and length constraints. Notice that by these results all main decision problems (i.e., membership, inclusion, and equivalence) for this class of pattern languages have an answer (summarized in Table \ref{tab:complexity-full} at the end of the final discussion).\label{section:patlang-reglen}

\section{Further Discussion}
\label{section:conclusion}
In this paper, we extended the research regarding decision problems on pattern languages. To get closer to an answer for the main open questions, e.g., the equivalence problem for erasing pattern languages or the larger alphabet cases of the terminal-free inclusion problem for non-erasing pattern languages, extending the research from \cite{Nowotka2024}, we now introduced length constraints as an alternative approach. With the results provided in this paper, we see that nearly all problems regarding pattern langauges with length constraints, except membership, are indeed undecidable (in particular the open cases for patterns without constraints). As mentioned at the end of Section \ref{section:patlang-len}, similar to \cite{Nowotka2024}, the decidability of the equivalence problem for non-erasing pattern languages with length constraints remains open.

By the help of some core ideas from \cite{Jiang1995,Freydenberger2010,BREMER201215,saarela:LIPIcs.ICALP.2020.140}, adapted constructions using length (and regular) constraints could be constructed to obtain these results. Regarding our first main result, Theorem \ref{theorem:patlang-len-erasing-equivalence-undecidable}, similar to \cite{Nowotka2024}, the main novelty for the proof is found in the way the constructed checker pattern $(\beta,\ell_\beta)$ does not produce \emph{too much} anymore, while still being able to handle everything from $(\alpha,\ell_\alpha)$ in a controlled manner. Length constraints helped to reduce the overhead in $(\beta,\ell_\beta)$ to enable the approaches from the inclusion problem proof to be used for the equivalence problem. Now, due to such a construction being shown to actually work for two different constraint types, e.g., regular and length constraints respectively, finding a way to adapt the existing proofs for pattern languages without any constraints in such a way that the same can be achieved here, poses a final challenge. If this can actually be done, the main open problem for pattern languages would be settled. So far, we do not know how likely such a construction is to be found. In addition, the use of length constraints also helped to achieve undecidability for decidable problems regarding patterns without constraints, i.e., the terminal-free inclusion and equivalence problems for erasing pattern languages. Hence, the expressiveness obtained using these constraints and, thus, needed to obtain undecidability of the equivalence problem for erasing pattern languages using the approach posed in this paper, might actually be unfeasible. Regarding the terminal-free inclusion problem in the non-erasing case, due to the undecidability of the binary case having already been shown for patterns without any constraints \cite{saarela:LIPIcs.ICALP.2020.140}, it is not unlikely that undecidability for larger alphabets can be achieved. Using the approach proposed in the current paper, we would need to somehow be able to restrict substitutions of some variables to single letters and distinct symbols in one case, and enable certain wildcards in any other case, similar to \cite{saarela:LIPIcs.ICALP.2020.140}, where a property of binary words is used to \emph{flip a switch}, determining where a word obtained from $\alpha$ has to be obtained in $\beta$. So far, this remains open.

Regarding the equivalence problem for non-erasing pattern languages for patterns with length constraints, we observe that this problem must be harder than the trivially decidable variant for non-erasing pattern languages without any constraints. This results from the fact that there are several NP-hard problems that are expressible solely by unions of linear diophantine equations. Some of which are, e.g., the subst sum problem, the knapsack problem, or integer programming~\cite{Karp1972}. Just to give a rough idea, we can easily embed, say, a positive integer subset sum instance $(S,T)$ with $S = \{S_1, ..., S_n\}$, for $S_i\in\N$ with $i\in[n]$, and $T\in\N$ to an instance of the equivalence problem for non-erasing pattern languages with length constraints by setting the pattern $\alpha = x_1...x_n$, for variables $x_1,...,x_n\in X$, defining $\ell_\alpha$ by setting $(x_i = S_i+1) \lor (x_i = 1)$, for $i\in[n]$, and setting $\sum_{i\in[n]}x_i = T + n$. Then, we set $\beta = y$, for a variable $y\in X$, with $\ell_\beta$ being defined by $y = T + n$. Then, $L_{NE}(\beta,\ell_\beta)$ only has words of length $T+n$ and $L_{NE}(\alpha,\ell_\alpha)$ only contains some words if $(S,T)$ actually has a solution, as otherwise we cannot find a valid substitution for $\alpha$. If it contains words, they have length $T+n$. So, both languages are equal if and only if there exists a solution for $(S,T)$. Here, disjunctions allowed to embed the choice of setting any $x_i$ to either length $1$ or length $S_i+1$. Hence, we obtain the following.

\begin{proposition}
    Let $(\alpha,\ell_\alpha),(\beta,\ell_\beta)\in\PatLen$ and let $\Sigma$ be any alphabet with at least $|\Sigma| \geq 1$. Deciding whether we have $L_{NE}(\alpha,\ell_\alpha) = L_{NE}(\beta,\ell_\beta)$ is NP-hard, even when we are restricted to terminal-free patterns $\alpha$ and $\beta$.
\end{proposition}

Notice that the previous result depended on the fact that we allow disjunctions in the given length constraints. The following result shows that at least in the case of unary alphabets, where we basically restrict the problem to a number setting, we can obtain NP-hardness without using disjunctions by a reduction from the the $\mathtt{3SAT}$ problem. Due to length restrictions, the proof of Proposition~\ref{prop:pat-len-equiv-ne-hard} can be found in Appendix~\ref{section:proof4}.

\begin{proposition}\label{prop:pat-len-equiv-ne-hard}
    Let $(\alpha,\ell_\alpha),(\beta,\ell_\beta)\in\PatLen$ where $\ell_\alpha$ and $\ell_\beta$ each use no disjunctions. Deciding whether we have $L_{NE}(\alpha,\ell_\alpha) = L_{NE}(\beta,\ell_\beta)$ is NP-hard for alphabets $\Sigma$ with $|\Sigma| = 1$, even when we are restricted to terminal-free patterns $\alpha$ and $\beta$.
\end{proposition}

For all other problems, we obtained an answer and notice that the most prominent open problems, i.e., the equivalence problem for erasing pattern languages and the inclusion problem for terminal-free non-erasing pattern languages for alphabets $\Sigma$ with $|\Sigma|\geq 3$ are indeed undecidable for patterns with length constraints. For patterns with regular and length constraints, we have even seen the undecidability of the equivalence problem of non-erasing pattern languages, concluding all open problems there. A final overview of the current state of results for decision problems on patterns over various constraints can be found in Table \ref{tab:complexity-full}. We propose the following open question for which we have no definite conjecture so far.

\begin{question}
    Given $(\alpha,\ell_\alpha),(\beta,\ell_\beta)\in\PatLen$, is it generally decidable to answer whether $L_{NE}(\alpha,\ell_\alpha) = L_{NE}(\beta,\ell_\beta)$?
\end{question}

\begin{table}[h!]
    \centering
    \setlength\tabcolsep{5.5pt} 
    \renewcommand{\arraystretch}{1.2}
    \begin{tabular}{ | C | C C | C C | C C | C C | }
        \hline
         & \multicolumn{2}{c|}{No Constraints} 
         & \multicolumn{2}{c|}{Len-Constraints} 
         & \multicolumn{2}{c|}{Reg-Constraints} 
         & \multicolumn{2}{c|}{RegLen-Constraints} \\
        \hline
         & Gen. & T.F.
         & Gen. & T.F. 
         & Gen. & T.F. 
         & Gen. & T.F. \\
        \hline \hline
        E ($\in$) & NPC & NPC & NPC & NPC& NPC & NPC & NPC & NPC \\
        \hline
        E ($\subseteq$)  & UD & NPC & UD & UD & UD & UD & UD & UD \\
        \hline
        E ($=$) & \textbf{Open} & NPC & UD & UD & UD & UD & UD & UD \\
        \hline
        NE ($\in$) & NPC & NPC & NPC & NPC & NPC & NPC & NPC & NPC \\
        \hline
        NE ($\subseteq$)  & UD & UD\tablefootnote{Undecidable in the binary case by \cite{saarela:LIPIcs.ICALP.2020.140}, open for larger alphabets.} & UD & UD & UD & UD & UD & UD \\
        \hline
        NE ($=$) & P & P & \textbf{Open} & \textbf{Open}\tablefootnote{Gen. and T.F. NP-hard as described above. An upper bound is unknown (may as well be undecidable).} & \textbf{Open} & \textbf{Open}\tablefootnote{Depending on regular languages representation, at least PSPACE-hard \cite{Nowotka2024}.} & UD & UD \\
        \hline
    \end{tabular}
    \caption{Summary of the current state of results regarding the main decision problems for pattern languages with different constraints (NPC = NP-Complete, UD = Undecidable, E = Erasing, NE = Non-Erasing, T.F. = Terminal Free, Gen. = General, ($\in$) means membership problem, ($\subseteq$) means inclusion problem, and ($=$) means equivalence problem.). The results for No Constraints and Regular Constraints are based on previous research mentioned in the introduction. The results for Length Constraints as well as Regular and Length Constraints summarize the results of this paper.}
    \label{tab:complexity-full}
\end{table}

\bibliography{patlang-len-constraints-main}

\appendix
\newpage

\section{Definition of Nondeterministic 2-Counter Automata without Input}
\label{section:defAut}
A \emph{nondeterministic 2-counter automaton without input} (see e.g. \cite{Iberra1978}) is a 4-tuple $A = (Q,\delta,q_0,F)$ which consists of a set of states $Q$, a transition function $\delta : Q \times \{0,1\}^2 \rightarrow \mathcal{P}(Q \times \{\-1,0,+1\}^2)$, an initial state $q_0\in Q$, and a set of accepting states $F \subseteq Q$. A \emph{configuration} of $A$ is defined as a triple $(q,m_1,m_2)\in Q\times\N\times\N$ in which $q$ indicates the current state and $m_1$ and $m_2$ indicate the contents of the first and second counter. We define the relation $\vdash_A$ on $Q\times\N\times\N$ by $\delta$ as follows. For two configurations $(p,m_1,m_2)$ and $(q,n_1,n_2)$ we say that $(p,m_1,m_2) \vdash_A (q,n_1,n_2)$ if and only if there exist $c_1,c_2\in\{0,1\}$ and $r_1,r_2\in\{-1,0,+1\}$ such that
\begin{enumerate}
	\item if $m_i = 0$ then $c_i = 0$, otherwise if $m_i > 0$, then $c_i = 1$, for $i\in\{1,2\}$,
	\item $n_i = m_i + r_i$ for $i\in\{1,2\}$,
	\item $(q,r_1,r_2)\in \delta(p,c_1,c_2)$, and
	\item we assume if $c_i = 0$ then $r_i \neq -1$ for $i\in\{1,2\}$.
\end{enumerate}
Essentially, the machine checks in every state whether the counters equal $0$ and then changes the value of each counter by at most one per transition before entering a new state. A \emph{computation} is a sequence of configurations. An \emph{accepting computation} of $A$ is a sequence $C_1,...,C_n\in (Q\times\N\times\N)^n$ with $C_1 = (q_0,0,0)$, $C_i \vdash_A C_{i+1}$ for all $i\in\{1,...,n-1\}$, and $C_n\in F\times\N\times\N$ for some $n\in\N$ with $n > 0$.

We \emph{encode} configurations of $A$ by assuming $Q = \{q_0,...,q_e\}$ for some $e\in\N$ and defining a function $\enc$ $:$ $Q\times\N\times\N \rightarrow \{0,\#\}^*$ by 
$$ \enc(q_i,m_1,m_2) := 0^{x+i}\#0^{c_1 + y_2\cdot m_1}\#0^{c_2 + y_2\cdot m_2} $$
for some numbers $x,c_1,y_2,c_1,y_2\in\N$. The values for these numbers depend on the construction of the respective proofs and are not specified here. Encodings of this kind are used to prove Theorem \ref{theorem:patlang-len-erasing-equivalence-undecidable} and Theorem \ref{theorem:patlang-reglen-nonerasing-equivalence-undecidability}. This is extended to encodings of computations by defining for every $n\geq 1$ and every sequence $C_1,...,C_n\in Q\times\N\times\N$
$$ \enc(C_1,...,C_n) := \#\#\ \enc(C_1)\ \#\#\ ...\ \#\#\ \enc(C_n)\ \#\#. $$

As mentioned before, for the proof of Theorem \ref{theorem:patlang-len-nonerasing-inclusion-tf-undecidability}, we use the notion of a specific universal Turing machine $U$ as it is used in \cite{BREMER201215}. Due to the lack of specifics needed to follow the proof, we omit its definition at this point and refer to \cite{BREMER201215} for details. Assume, however, that the encoding function $\enc$ is also defined for computations and configurations of $U$, resulting in words of other structure, though.

If $A$ is either a nondeterministic 2-counter automaton without input or the universal Turing machine $U$ mentioned before, define the set of encodings of accepting computations 
$$\mathtt{ValC}(A) := \{\enc(C_1,...,C_n)\ |\ C_1,...,C_n \text{ is an accepting computation of } A\}$$
and let $\mathtt{InvalC}(A) = \{0,\#\}^*\setminus \mathtt{ValC}(A)$. 
The emptiness problem for deterministic 2-counter-automata is undecidable (cf. e.g. \cite{Iberra1978,Minsky1961}), thus it is also undecidable whether a nondeterministic 2-counter automaton without input has an accepting computation \cite{Freydenberger2010, Jiang1995}. That the emptiness problem for universal Turing machines is undecidable is a known fact.
\newpage

\section{The Undecidability of the Equivalence Problem for Erasing Pattern Languages with Length Constraints}
\label{section:proof1}
This section is dedicated to show Theorem \ref{theorem:patlang-len-erasing-equivalence-undecidable}. The idea is roughly based on the idea for the proof of undecidability of the inclusion problem for erasing pattern languages without any constraints given in \cite{Jiang1995} or \cite{Freydenberger2010} and extended in a such way that it works for the equivalence problem for erasing pattern languages with length constraints. 

The structure of this section is as follows: We continue with an intuitive explanation of the proof \cite{Freydenberger2010}, as our construction uses certain key takeaways and results from that proof. After that, we provide a formal definition of our adapted construction using length constraints. Then, we show a series of results obtained by that construction. Finally, these results are put together to show the correctness of the reduction in a concise manner. In the end, a short note on how to extend this to any alphabet size larger or equal to $3$ is provided, concluding this section.

\textbf{Summary of the result in \cite{Freydenberger2010}:} In \cite{Freydenberger2010}, the authors reduce the question of whether some non-deterministic 2-counter automaton without input $A$ has some accepting computation to the question whether the erasing pattern languages of one pattern is a subset of the erasing pattern languages of some other pattern. For that, in the binary case, given some non-deterministic 2-counter automaton without input $A$, two patterns $\alpha$ and $\beta$ are constructed such that $L_E(\alpha) \subseteq L_E(\beta)$ if and only if $A$ has no accepting computation. For that, the pattern $\alpha$ is constructed in such a way that all words in $L_E(\alpha)$ are of the form $vv\ \#^4 v\ x\ v\ x\ v\ \#^4\ vuv$ with $v = 0\#^30$, $u = 0\#^20$, and $x,y\in\Sigma^*$, simply by setting
$$\alpha = vv\ \#^4\ v\ x\ v\ y\ x\ \#^4\ vuv$$
for variables $x,y\in X$.
The pattern $\beta$ is then defined by
$$ \beta = (x_1)^2 ... (x_\mu)^2\ \#^4\ \hat{\beta}_1 ... \hat{\beta}_\mu\ \#^4\ \ddot{\beta}_1 ... \ddot{\beta}_\mu$$
for some constant $\mu\in\N$, variables $x_1,...,x_\mu\in X$, and terminal free patterns $\hat{\beta}_1,...,\hat{\beta_\mu},\ddot{\beta}_1...\ddot{\beta}_\mu$ such that for all $i\in[\mu]$ we have $\hat{\beta}_i = x_i\ \gamma_i\ x_i\ \delta_i\ x_i$ and $\ddot{\beta}_i = x_i\ \eta_i\ x_i$ for terminal free patterns $\gamma_i,\delta_i\in X^*$ and $\eta_i = z_i(\hat{z}_i)z_i$ for variables $z_i,\hat{z_i}\in X$. It is important to mention that all variables occurring in $\gamma_i$, $\delta_i$ and $\eta_i$ are only shared between these three patterns and occur nowhere else, for each $i\in[\mu]$. For each $i\in[\mu]$, the triple $(\gamma_i,\delta_i,\eta_i)$ is called a \emph{predicate} $\pi_i$ which is \emph{satisfied} if and only if for any substitution $h$ on $\alpha$ we have that there exists some homomorphism $\tau$ with $\tau(\gamma_i) = h(x)$, $\tau(\delta_i) = h(y)$, and $\tau(\eta_i) = u = 0\#^20$. Notice that, from the sentence before, we know that each predicate basically works on its own set of variables, disjoint from any other predicate and the rest of the pattern. It is shown that, for any substitution $h$, in order to obtain $h(\alpha)\in L_E(\beta)$, $h$ must satisfy some predicate $\pi_i$ for $i\in[\mu]$. The patterns $\delta_i$ and $\gamma_i$ are then designed in such a way that $h$ can only satisfy some predicate if we do not have that, first, $h(x)$ is the encoding of a valid computation on $A$, and, second, $h(y)\in\{0\}^{\geq n}$ for $n$ being the sum of the longest three separate factors of $0$'s in $h(x)$. By that, it is obtained that $L_E(\alpha)\subseteq L_E(\beta)$ if and only if $A$ has no valid computation as otherwise such a substitution $h$ satisfying no predicate can always be found. 

As a key takeaway, note in particular that any substitution $h$ satisfying some predicate $\pi_i$ results in exactly some substitution $h'$ for $\beta$ with $h(\alpha) = h'(\beta)$ such that $x_i$ is substituted to $v$, $\eta_i$ substituted to $u$ (resulting in $z_i$ and $\hat{z_i}$ being substituted to $0$ and $\#$ respectively), $\hat{\beta}_i$ being substituted to some nonempty word $vx'vy'v$ for $x',y'\in\Sigma^*$ and all other parts of $\beta$ being substituted to the empty word. In order for the predicates $(\gamma_i,\delta_i,\eta_i)$ to work, it is necessary that $\eta_i$ contains two variables $z_i$ and $\hat{z}_i$ which are substituted to the letters $0$ and $\#$. These are used to construct the structural properties forced in $\gamma_i$ and $\delta_i$ to obtain words that have factors that cannot be part of any encoding of some valid computation. We will make use of that to be able to reuse the same predicates $\pi_1$ to $\pi_\mu$ in the construction of this proof while also adding new ones for new cases that arise.

The problem regarding the equivalence problem is the following. The pattern $\beta$ results in $L_E(\beta)$ having many more structurally different words than the ones in $L_E(\alpha)$ (e.g., by substituting $x_1$ to $x_\mu$ all with different and very long words). Hence, in the case of $A$ having no valid computation, $L_E(\beta)$ is actually always a proper superset of $L_E(\alpha)$. In contrast to the proof in \cite{Freydenberger2010}, we can make use of length constraints to restrict that. The main idea of the proof of Theorem \ref{theorem:patlang-len-erasing-equivalence-undecidable} is that we can construct two patterns with length constraints $(\alpha,\ell_\alpha)$ and $(\beta,\ell_\beta)$ such that $L_E(\beta,\ell_\beta)$ only contains words of a structure that can also be obtained in $L_E(\alpha,\ell_\alpha)$, effectively restricting the language $L_E(\beta,\ell_\beta)$ significantly. While the main idea of this proof stays the same compared to, e.g., \cite{Jiang1995} or \cite{Freydenberger2010}, we now obtain that $L_E(\alpha,\ell_\alpha) = L_E(\beta,\ell_\beta)$ if and only if $A$ has no accepting computation, otherwise we now have $L_E(\alpha,\ell_\alpha) \supset L_E(\beta,\ell_\beta)$. Whether such a restricted form of $\beta$ together with a fitting $\alpha$ can be constructed without any additional constraints, remains an open problem, but would show the undecidability of the equivalence problem of pattern languages without any constraints. In order for that to work with length constraints, the patterns were adapted significantly. So much so that this proof even works with terminal-free patterns, showing undecidability of a case that is decidable for patterns without any constraints. We continue, now, with the reduction of the emptiness problem for non-deterministic 2-counter automata without input to the equivalence problem of erasing pattern languages with length constraints. From now on, $\alpha$ and $\beta$ are redefined by the construction below and do not refer to $\alpha$ and $\beta$ from the the explanations above if not mentioned otherwise.

\textbf{Our Construction:} We continue with the construction to show Theorem \ref{theorem:patlang-len-erasing-equivalence-undecidable}. Let $A = (Q,\delta,q_0,F)$ be some non-deterministic 2-counter automaton without input. We construct two patterns with length constraints $(\alpha,\ell_\alpha),(\beta,\ell_\beta)\in\PatLen$ such that $L_E(\alpha,\ell_\alpha) = L_E(\beta,\ell_\beta)$ if and only if $\mathtt{ValC}(A) = \emptyset$. We start with the binary case and assume that $\Sigma = \{0,\#\}$. First, we construct $(\alpha,\ell_\alpha)$. We set the pattern $\alpha$ to
$$ \alpha = y_1\ x_v\alpha_1x_v\alpha_2x_v\ y_2\ x_vzz'z'zx_v $$
for new and independent variables $y_1,y_2,x_v,\alpha_1,\alpha_2,z,z'\in X$. We define the following length constraints $\ell_\alpha$:
\begin{itemize}
    \item $x_v = 5$
    \item $z = 1$
    \item $z' = 1$
\end{itemize}
Notice that any word in $L_E(\alpha,\ell_\alpha)$ has a suffix of fixed length with the form $v u v$ where $v\in\Sigma^5$ and $u\in\{0000, \#\#\#\#, 0\#\#0, \#00\#\}$. In comparison to the construction in \cite{Freydenberger2010} that made use of a specific prefix and a specific suffix to enforce predicate choice, only this suffix suffices for this purpose in this construction.

Next, we construct $(\beta,\ell_\beta)$. We set the pattern $\beta$ to
$$ \beta = y_1'\ \hat{\beta}_1\ \hat{\beta}_2\ \cdots\ \hat{\beta}_{\mu+1}\ y_2'\ \ddot{\beta}_1\ \ddot{\beta}_2\ \cdots\ \ddot{\beta}_{\mu+1}$$
for new and independent variables $y_1',y_2'\in X$ and terminal-free patterns $\hat{\beta}_i,\ddot{\beta}_i\in X^*$ for $i\in[\mu+1]$ with
\begin{itemize}
    \item $\hat{\beta}_i = x_i\ \gamma_i\ x_i\ \delta_i\ x_i$,
    \item $\ddot{\beta}_i = x_i\ \eta_i\ x_i$
\end{itemize}
for new and independent variables $x_i\in X$ and terminal-free patterns $\gamma_i,\delta_i,\eta_i\in X^*$. $\mu$ is a number to be defined later. We assume that
\begin{itemize}
    \item $\eta_i = z_iz_i'z_i'z_i$ and $z_i \neq z_i'$ for $i\in[\mu]$,
    \item $\eta_{\mu+1} = (z_{\mu+1})^4$,
    \item $\var(\gamma_i\delta_i\eta_i)\cap\var(\gamma_j\delta_j\eta_j) = \emptyset$ if $i \neq j$ for $i,j\in[\mu+1]$, and
    \item $x_k,y_1',y_2'\notin\var(\gamma_i\delta_i\eta_i)$ for all $i,k\in[\mu+1]$
\end{itemize}
for new and independent variables $z_i,z_i'\in X$.
Hence, all variables inside $\var(\gamma_i\delta_i\eta_i)$ do not occur anywhere else but in these factors. The length constraint $\ell_\beta$ is defined by the following and only the following system (no additional constraints will be added later on).
\begin{enumerate}
    \item $\sum_{i=1}^{\mu+1}z_i = 1$
    \item $\sum_{i=1}^{\mu}z_i' \leq 1$
    \item $\sum_{i=1}^{\mu+1}x_i = 5$
    \item $x_i - 5z_i = 0$ for all $i\in[\mu+1]$
    \item $z_i-z_i' = 0$ for all $i\in[\mu]$
\end{enumerate}

\textbf{Properties of the constructed pattern languages with length constraints:} We continue with a collection of results obtained by the structure of these patterns. We start by showing $L_E(\beta,\ell_\beta) \subseteq L_E(\alpha,\ell_\alpha)$ and then continue how we use the predicates defined in \cite{Freydenberger2010} to obtain a valid reduction in this adapted case.

First, the given length-constraints result in some very specific properties of words in $L_E(\beta,\ell_\beta)$. We consider the following properties to be the most important ones for the rest of this proof. They show that all $\ell_\beta$-valid substitutions $h$ effectively lead to a choice of a single $i\in[\mu+1]$ for which we have that $h(\ddot{\beta}_i) \neq \varepsilon$ while for all other $j\in[\mu+1]\setminus\{i\}$ we have that $h(\ddot{\beta}_j) = \varepsilon$. Out of that, the property $L_E(\beta,\ell_\beta) \subseteq L_E(\alpha,\ell_\beta)$ can be obtained, which is the key difference to \cite{Freydenberger2010} mentioned before. We begin with properties of all $\ell_\beta$-valid substitutions.

\begin{lemma}\label{lemma:patlen-e-equiv-undec-beta-length-property}
    Let $h\in H_{\ell_\beta}$ be a $\ell_\beta$-valid substitution. There exists some $i\in[\mu+1]$ such that 
    \begin{enumerate}
        \item $h(x_i) = v$ for some $v\in\Sigma^*$ with $|v| = 5$,
        \item $h(z_i) \in \Sigma$, i.e., $|h(z_i)| = 1$,
        \item $h(z_i') \in \Sigma$, i.e., $|h(z_i')| = 1$ if $i\neq \mu+1$,
        \item $h(x_j) = h(z_j) = h(z_{j'}') = \varepsilon$ for all $j\in[\mu+1]$ and $j'\in[\mu]$ with $i\neq j$ and $i\neq j'$.
    \end{enumerate}
\end{lemma}
\begin{proof}
    We get (2) by property 1 of $\ell_\beta$. By property 4 of $\ell_\beta$ we then obtain (1). By property 5 of $\ell_\beta$ we obtain (3). By the previous three attributes and properties 1, 2, 3, 4, and 5 combined, (4) immediately follows.
\end{proof}

By Lemma \ref{lemma:patlen-e-equiv-undec-beta-length-property}, we obtain the following structure of all words in $L_E(\beta,\ell_\beta)$.

\begin{lemma}\label{lemma:patlen-e-equiv-undec-beta-subst-property}
    Let $h\in H_{\ell_\beta}$. Then there exists some $v\in\Sigma^*$ with $|v| = 5$ and potentially equal $\ta,\tb\in\Sigma$ such that $h(\beta) = h(y_1'\hat{\beta}_1...\hat{\beta}_{i-1})\ v\ h(\gamma_i)\ v\ h(\delta_i)\ v\ h(\hat{\beta}_{i+1}...\hat{\beta}_{\mu+1}y_2')\ v\ta\tb\tb\ta v.$
\end{lemma}
\begin{proof}
    By Lemma, \ref{lemma:patlen-e-equiv-undec-beta-length-property} we immediately obtain that there exists some $i\in[\mu+1]$ such that $h(\ddot{\beta}_i) = h(x_iz_iz_i'z_i'z_ix_i) = v\ta\tb\tb\ta v$ if $i \leq \mu$, or, if $i = \mu+1$, then $h(\ddot{\beta}_i) = h(x_iz_i^4x_i) = v\ta^4v$ for some $v\in\Sigma^*$ with $|v| = 5$ and $\ta,\tb\in\Sigma^*$. Also, by Lemma \ref{lemma:patlen-e-equiv-undec-beta-length-property}, we obtain that for all other $j\in[\mu+1]$ with $j\neq i$ we have $h(\ddot{\beta}_i) = \varepsilon$. By the fact that $h(x_i) = v$, we also obtain that $h(\hat{\beta}_i) = h(x_i \gamma_i x_i \delta_i x_i) = v\ h(\gamma_i)\ v\ h(\delta_i)\ v$. So, applying $h$ to $\beta$ results in
    \begin{align*}
        h(\beta) &= h(y_1'\ \hat{\beta}_1\ \cdots\ \hat{\beta}_{\mu+1}\ y_2'\ \ddot{\beta}_1\ \cdots\ \ddot{\beta}_{\mu+1}) \\
                 &= h(y_1'\ \hat{\beta}_1\ ...\ \hat{\beta}_i)\ h(\hat{\beta}_i)\ h(\hat{\beta}_{i+1}\ ...\ \hat{\beta}_{\mu+1}\ y_2')\ h(\ddot{\beta}_1\ ...\ \ddot{\beta}_{\mu+1}) \\
                 &= h(y_1'\ \hat{\beta}_1 ... \hat{\beta}_i)\ v\ h(\gamma_i)\ v\ h(\delta_i)\ v\ h(\hat{\beta}_{i+1}...\hat{\beta}_{\mu+1}\ y_2')\ v\ta\tb\tb\ta v
    \end{align*}
    by which the lemma holds.
\end{proof}

Using Lemma \ref{lemma:patlen-e-equiv-undec-beta-subst-property}, we can show that the language $L_E(\beta,\ell_\beta)$ is contained in $L_E(\alpha,\ell_\alpha)$. Obtaining this property effectively allows this proof technique to be used for the equivalence problem instead of only the inclusion problem for pattern languages with length constraints.

\begin{proposition}\label{prop:proof1-langbeta-in-langalpha}
    We have $L_E(\beta,\ell_\beta)\subseteq L_E(\alpha,\ell_\alpha)$.
\end{proposition}
\begin{proof}
    Let $w\in L_E(\beta,\ell_\beta)$. Then, there exists some $h\in H_{\ell_\beta}$ such that $h(\beta) = w$. We know by Lemma \ref{lemma:patlen-e-equiv-undec-beta-subst-property} that there exists some $v\in\Sigma^*$ with $|v| = 5$ and $\ta,\tb\in\Sigma$ such that
    $$h(\beta) = h(y_1'\hat{\beta}_1\ ...\ \hat{\beta}_{i-1})\ v\ h(\gamma_i)\ v\ h(\delta_i)\ v\ h(\hat{\beta}_{i+1}\ ...\ \hat{\beta}_{\mu+1}y_2')\ v\ta\tb\tb\ta v.$$

    Now, let $h'\in H$ and set $h'(y_1) = h(y_1'\hat{\beta}_1...\hat{\beta}_{i-1})$, $h'(y_2) = h(\hat{\beta}_{i+1}...\hat{\beta}_{\mu+1}y_2')$, $h'(x_v) = v$, $h'(\alpha_1) = h(\gamma_i)$, $h'(\alpha_2) = h(\delta_i)$, $h'(z) = \ta$, and $h'(z') = \tb$. Then, as $|h'(x_v)| = |v| = 5$ and $|h'(z)| = |h(z_i)| = 1 = |h(z_i')| = |h'(z')|$ we know that $h'$ is $\ell_\alpha$-valid. We obtain
    \begin{align*}
        h'(\alpha) &= h'(y_1\ x_v\alpha_1x_v\alpha_2x_v\ y_2\ x_vzz'z'zx_v) \\
                   &= h(y_1'\hat{\beta}_1\ ...\ \hat{\beta}_{i-1})\ v\ h(\gamma_i)\ v\ h(\delta_i)\ v\ h(\hat{\beta}_{i+1}...\hat{\beta}_{\mu+1}y_2')\ v\ta\tb\tb\ta v \\
                   &= h(\beta) = w.
    \end{align*}
    So, there exists some $h'\in H_{\ell_\alpha}$ such that $h'(\alpha) = w$, thus $w\in L_E(\alpha,\ell_\alpha)$.
\end{proof}

Now, we have to show that words $w\in L_E(\alpha,\ell_\alpha)$ are also in $L_E(\beta,\ell_\beta)$ if and only if certain conditions are fulfilled. In particular, we use the previously mentioned notion of predicates from \cite{Freydenberger2010} and show that for some $w\in L_E(\alpha,\ell_\alpha)$ we have $w\in L_E(\beta,\ell_\beta)$ if either some trivial property is fulfilled or at least one predicate is satisfied. As we will see, this will only not be the case if, for some $h\in H_{\ell_\alpha}$, we have that $h(x)\in\ValC(A)$ and $h(y)\in\{0\}^*$ for $|h(y)|$ being at least as long as the sum of the longest three factors of $0$'s in $h(x)$. We continue with a collection of results that, in the end, will be used to conclude the correctness of the reduction in a concise manner.

As in \cite{Freydenberger2010}, we interpret each triple $(\gamma_i,\delta_i,\eta_i)$ as a predicate $\pi_i : H_{\ell_\alpha} \rightarrow \{true,false\}$. We say that a substitution $h\in H_{\ell_\alpha}$ \emph{satisfies} the predicate $\pi_i$, i.e., $\pi_i(h) = true$, if there exists some morphism $\tau : \var(\gamma_i\delta_i\eta_i)^* \rightarrow \Sigma^*$ with $\tau(\gamma_i) = h(\alpha_1)$, $\tau(\delta_i) = h(\alpha_2)$ and $\tau(\eta_i) = h(zz'z'z)$. Otherwise, we have that $\pi_i$ is not satisfied, i.e., $\pi_i(h) = false$. 

From now on, we assume the predicates $\pi_1$ to $\pi_\mu$ to be exactly constructed as in \cite{Freydenberger2010} by setting $\gamma_i$ and $\delta_i$ accordingly for each $i\in[\mu]$, and by using the same $z_i$ and $z_i'$ variables to encode the structural properties of encodings of invalid computations for $A$ in the case of $z$ being substituted to $0$ and $z_i'$ being substituted to $\#$, respectively, or vice versa. Due to the extent of their construction, we refer to \cite{Freydenberger2010} for the definitions of $(\gamma_i,\delta_i,\eta_i)$ for $i\in[\mu]$. We, however, mention the following conclusive result from \cite{Freydenberger2010}, adapted to our construction.

\begin{proposition}\label{prop:proof1sourceresult}\cite{Freydenberger2010}
    Let $h\in H_{\ell_\alpha}$ and assume $h(z) = 0$ and $h(z') = \#$. Then, $h$ satisfies no predicate $\pi_i$ for $i\in[\mu]$ if and only if $h(\alpha_1)\in\ValC(A)$ and $h(\alpha_2)\in\{0\}^{\geq n}$ for $n$ being the length of the three longest non-overlapping factors of $0$'s in $h(\alpha_1)$ summed up.
\end{proposition}

Indeed, this follows from the fact the predicates are defined independent of structural variables like $x_v$, $y_1$ or $y_2$, and by the fact that $h(z) = 0$ and $h(z') = \#'$, resulting in the factor $h(zz'z'z) = 0\#\#0$. By that, $h$ satisfies some predicate $\pi_i$, for $i\in[\mu]$, if and only if there exists some homomorphism $\tau$ such that $h(\alpha_1) = \tau(\gamma_i)$, $h(\alpha_2) = \tau(\delta_i)$ and $\tau(\eta_i) = 0\#\#0$, following the exact structure from \cite{Freydenberger2010} as mentioned in the beginning of this section. As Freydenberger and Reidenbach show that the predicates $\pi_1$ to $\pi_\mu$ are only not satisfied for the properties of $h$ mentioned in proposition \ref{prop:proof1sourceresult}, we obtain that result for our construction in that case as well.

In summary, these predicates allow for the differentiation of substitutions that contain encoded accepting computations of $A$ and those that do not. However, we have to assume $h(z) \neq h(z')$ to fix positions of the letters $0$ and $\#$. So, before showing that these predicates can be embedded in our construction, we need to exclude the case of $h(z) = h(z')$ by forcing such substitutions to result in words that are always in $L_E(\beta,\ell_\beta)$. We define the additional predicate $\pi_{\mu+1}$ by using the previously defined $\eta_{\mu+1} = (z_{\mu+1})^4$ and setting
\begin{itemize}
    \item $\gamma_{\mu+1} = y_{\mu+1}$, and
    \item $\delta_{\mu+1} = y_{\mu+1}'$
\end{itemize}
for new and independent variables $y_{\mu+1}$ and $y_{\mu+1}'$. We can see that, if for some substitution $h\in H_{\ell_\alpha}$ we have $h(z) = h(z')$, then the predicate $\pi_{\mu+1}$ is always satisfied. This results in $h(\alpha)\in L_E(\beta,\ell_\beta)$ in these cases.

\begin{lemma}\label{lemma:patlen-e-equiv-undec-equal-z-always-in-beta}
    Let $h\in H_{\ell_\alpha}$. We have $h(z) = h(z')$ if and only if $\pi_{\mu+1}(h) = true$.
\end{lemma}
\begin{proof}
    First assume $h(z) = h(z')$. Then
    \begin{align*}
        h(\alpha) &= h(y_1x_v\alpha_1x_v\alpha_2x_vy_2x_vzz'z'zx_v) \\
                  &= h(y_1x_v\alpha_1x_v\alpha_2x_vy_2x_vz^4x_v).
    \end{align*}
    Select $\tau\in H$ such that $\tau(z_{\mu+1}) = h(z)$, $\tau(y_{\mu+1}) = h(\alpha_1)$, and $\tau(y_{\mu+1}') = h(\alpha_2)$. Then $\tau(\gamma_{\mu+1}) = \tau(y_{\mu+1}) = h(\alpha_1)$, $\tau(\delta_{\mu+1}) = \tau(y_{\mu+1}') = h(\alpha_2)$, and $\tau(\eta_{\mu+1}) = \tau(z_{\mu+1}^4) = \tau(z_{\mu+1})^4 = h(z)^4 = h(zz'z'z)$. So, by the existence of $\tau$ we know that $h$ satisfies $\pi_{\mu+1}$.

    For the other direction, assume that $h$ satisfies $\pi_{\mu+1}$. Then there exists some morphism $\tau$ with $\tau(\gamma_{\mu+1}) = h(\alpha_1)$, $\tau(\delta_{\mu+1}) = h(\alpha_2)$, and $\tau(\eta_{\mu+1}) = h(zz'z'z)$. We know $\eta_{\mu+1} = z_{\mu+1}^4$. We also know $|h(z)| = |h(z')| = 1$. Hence, $|h(zz'z'z)| = 4$ and by that $|\tau(\eta_{\mu+1})| = |\tau(z_{\mu+1}^4)| = 4$. By that, we have $|\tau(z_{\mu+1})| = 1$, resulting in $\tau(z_{\mu+1}) = h(z) = h(z')$.
\end{proof}

Using this result, we can immediately conclude that all substitutions $h\in H_{\ell_\alpha}$ with $h(z) = h(z')$ result in $h(\alpha) \in L_E(\beta,\ell_\beta)$, as mentioned before.

\begin{corollary}\label{corollary:proof1-z-equal-z}
    Let $h\in H_{\ell_\alpha}$. If $h(z) = h(z')$, then $h(\alpha) \in L_E(\beta,\ell_\beta)$.
\end{corollary}
\begin{proof}
    By Lemma \ref{lemma:patlen-e-equiv-undec-equal-z-always-in-beta}, we know $h$ satisfies $\pi_{\mu+1}$. So, we find a morphism $\tau$ such that $\tau(\gamma_{\mu+1}) = \tau(y_{\mu+1}) = h(\alpha_1)$, $\tau(\delta_{\mu+1}) = \tau(y_{\mu+1}') = h(\alpha_2)$, and $\tau(\eta_{\mu+1}) = h(zz'z'z)$. Extend that to a substitution $\tau'$ by $\tau'(y_{\mu+1}) = h(\alpha_1)$, $\tau'(y_{\mu+1}') = h(\alpha_2)$, $\tau'(x_{\mu+1}) = h(x_v)$, $\tau'(z_{\mu+1}) = h(z) = h(z')$, $\tau'(y_1') = h(y_1)$, $\tau'(y_2') = h(y_2)$, and $\tau'(x) = \varepsilon$ for all other variables $x\in X$. Then, as $|\tau'(x_{\mu+1})| = |h(x_v)| = 5$, $|\tau'(z_{\mu+1})| = |h(z)| = |h(z')| = 1$, as well as $\tau'(x_i) = \tau'(z_i) = \tau'(z_i') = \varepsilon$ is satisfied, we know that $\tau'$ must be $\ell_\beta$-valid, i.e., $\tau'\in H_{\ell_\beta}$. We get

    \begin{align*}
        \tau'(\beta) &= \tau'(y_1'\ \hat{\beta}_1\ ...\ \hat{\beta}_{\mu+1}\ y_2'\ \ddot{\beta}_1\ ...\ \ddot{\beta}_{\mu+1}) \\
                     &= \tau'(y_1'\hat{\beta}_{\mu+1}y_2'\ddot{\beta}_{\mu+1}) \\
                     &= \tau'(y_1'x_{\mu+1}\gamma_{\mu+1}x_{u+1}\delta_{\mu+1}x_{\mu+1}y_2')\tau'(x_{\mu+1}z_{\mu+1}^4x_{\mu+1}) \\
                     &= h(y_1x_v)\tau'(y_{\mu+1})h(x_v)\tau'(y_{\mu+1}')h(x_vy_2)h(x_vzz'z'zx_v) \\
                     &= h(y_1x_v)h(\alpha_1)h(x_v)h(\alpha_2)h(x_vy_2x_vzz'z'zx_v) \\
                     &= h(\alpha).
    \end{align*}
    This concludes this lemma.
\end{proof}

By that, we only have to consider substitutions $h\in H_{\ell_\alpha}$ with $h(z) \neq h(z')$ to obtain words that are potentially not in $L_E(\beta,\ell_\beta)$. Without having to consider the specific construction of the predicates $\pi_1$ to $\pi_\mu$, we obtain the following result similar to \cite{Freydenberger2010}. It essentially shows that if any predicate $\pi_i$ for $i\in[\mu+1]$ is satisfied for some substitution $h$, then we have $h(\alpha) \in L_E(\beta,\ell_\beta)$. 

\begin{lemma}\label{lemma:palen-e-equiv-undec-predicate-language}
    Let $h\in H_{\ell_\alpha}$. If $\pi_i(h) = true$ for some $i\in[\mu+1]$, then $h(\alpha)\in L_E(\beta,\ell_\beta)$.
\end{lemma}
\begin{proof}
    Let $h\in H_{\ell_\alpha}$. If $i = \mu+1$, the result follows immedaitely by Corollary \ref{corollary:proof1-z-equal-z}. Hence, assume there exists some $i\in[\mu]$ such that $\pi_i$ is satisfied. Then we know there exists some morphism $\tau$ such that $\tau(\gamma_i) = h(\alpha_1)$, $\tau(\delta_i) = h(\alpha_2)$ and $\tau(\eta_i) = \tau(z_iz'_iz'_iz_i) = h(zz'z'z)$. Create a substitution $h'$ such that $h'(\gamma_i) = \tau(\gamma_i)$, $h'(\delta_i) = \tau(\delta_i)$ and $h'(\eta_i) = h(zz'z'z)$. As all variables in the predicate $\pi_i$ only appear in that predicate, we can assume such a substitution to exist. Additionally, set $h'(x_i) = h(x_v)$, $h'(y_1') = h(y_1)$ as well as $h'(y_2') = h(y_2)$. For all other variables $x\in\var(\beta)$ we set $h'(x) = \varepsilon$. By the length constraints $\ell_\beta$ we get that $h'$ is $\ell_\beta$-valid. We obtain that
    \begin{align*} 
        h'(\beta) &= h'(y_1'\ \hat{\beta}_1\ \hat{\beta}_2\ \cdots\ \hat{\beta}_{\mu+1}\ y_2'\ \ddot{\beta}_1\ \ddot{\beta}_2\ \cdots\ \ddot{\beta}_{\mu+1}) \\
                 &= h'(y_1')h'(\hat{\beta}_i)h'(y_2')h'(\ddot{\beta}_i) \\
                 &= h(y_1)h'(x_i\gamma_ix_i\delta_ix_i)h(y_2)h'(x_iz_iz_i'z_i'z_ix_i) \\
                 &= h(y_1x_v)h(\alpha_1)h(x_v)h(\alpha_2)h(x_vy_2)h(h_vzz'z'zx_v) \\
                 &= h(y_1x_v\alpha_1x_v\alpha_2x_vy_2x_vzz'z'zx_v) = h(\alpha).
    \end{align*}
    Thus, $h(\alpha)\in L_E(\beta,\ell_\beta)$, which concludes this lemma.
\end{proof}

So, using Lemma \ref{lemma:palen-e-equiv-undec-predicate-language} and following the results in \cite{Freydenberger2010}, we can conclude that if some predicate $\pi_i$ for $i\in[\mu+1]$ is satisfied by some $h\in H_{\ell_\alpha}$, then we cannot have that $h(\alpha_1)$ is the encoding of some valid computation of $A$ while $h(\alpha_2)$ is a word containing only $0$'s that is long enough. 

Showing the other direction, i.e., if $h(\alpha)\in L_E(\beta,\ell_\beta)$ then some predicate $\pi_i$ must be satisfied, would conclude this proof immediately, but is not generally possible and not generally necessary, here. Indeed, it suffices to show that there exist substitutions $h\in H_{\ell_\alpha}$ with a certain structure that force the use of some predicate $\pi_i$, $i\in[\mu]$, if we want $h(\alpha)\in L_E(\beta,\ell_\beta)$. This is then used to construct a counter example to show $L_E(\alpha,\ell_\alpha) \neq L_E(\beta,\ell_\beta)$ if $A$ actually has an accepting computation. The other case, i.e., $L_E(\alpha,\ell_\alpha) = L_E(\beta,\ell_\beta)$ if $A$ has no accepting computation, can then be shown afterwards.

Actually, this specific structure results in words of the form used in \cite{Freydenberger2010} and mentioned above. We begin by proving that this structure for substitutions of $\alpha$, i.e., substituting $x_v$ with $0\#\#\#0$, $z$ with $0$, and $z'$ with $\#$, results in us being forced to satisfy some predicate $\pi_i$, $i\in[\mu]$, to obtain $h(\alpha)\in L_E(\beta,\ell_\beta)$. Hence, the following result provides the other direction of Lemma \ref{lemma:palen-e-equiv-undec-predicate-language} in a restricted case.

\begin{lemma}\label{lemma:patlen-e-equiv-undec-language-to-predicate}
    Let $h\in H_{\ell_\alpha}$ such that $h(x_v) = 0\#\#\#0$, $h(z) = 0$, $h(z') = \#$, and $h(y_1) = h(y_2) = \varepsilon$. Then we know that $h(\alpha) \in L_E(\beta,\ell_\beta)$ if and only if there exists some $i\in[\mu]$ such that $h$ satisfies the predicate $\pi_i$.
\end{lemma}
\begin{proof}
    The if direction follows by Lemma \ref{lemma:palen-e-equiv-undec-predicate-language}. For the other direction we know, as $h(\alpha) \in L_E(\beta,\ell_\beta)$, that there exists some substitution $h'\in H_{\ell_\beta}$ such that $h(\alpha) = h'(\beta)$. First, by the construction of $\pi_1$ and $\pi_2$ in \cite[Theorem 2]{Freydenberger2010}, we know that if $h(\alpha_1)$ has an occurrence of the factor $\#^3$ or if $h(\alpha_2)$ has an occurrence of the letter $\#$, that then one of these two predicates is satisfied. Assume that $h$ does not satisfy $\pi_1$ or $\pi_2$. Now, due to Lemma \ref{lemma:patlen-e-equiv-undec-beta-subst-property}, we know that there exists some $i\in[\mu]$ such that
    $$h(\beta) = h'(y_1'\hat{\beta}_1...\hat{\beta}_{i-1})\ v\ h(\gamma_i)\ v\ h(\delta_i)\ v\ h(\hat{\beta}_{i+1}...\hat{\beta}_{\mu+1}y_2')\ v0\#\#0 v$$ for $v = 0\#\#\#0$ as
    $$h(\alpha) = vh(\alpha_1)vh(\alpha_2)v\ v0\#\#0v.$$
    In particular, we obtain that $h(\alpha)$ and $h'(\beta)$ share the suffix $v0\#\#0v$. That leaves us with
    \begin{align*}
        h(\alpha) &= v h(\alpha_1) v h(\alpha_2) v\ v0\#\#0v \\
                  &= h'(y_1'\hat{\beta}_1...\hat{\beta}_{i-1})\ v\ h(\gamma_i)\ v\ h(\delta_i)\ v\ h(\hat{\beta}_{i+1}...\hat{\beta}_{\mu+1}y_2')\ v0\#\#0v \\
                  &= h'(\beta)
    \end{align*}
    As $h(\alpha_1)$ does not contain $\#^3$ and $h(\alpha_2)$ does not contain any $\#$, we know that the occurrences of $v$ must align in the substitutions. Hence, we get that $h(\alpha_1) = h'(\gamma_i)$ and $h(\alpha_2) = h'(\delta_i)$ must be the case. As $h(zz'z'z) = 0\#\#0 = h'(\eta_i)$ is also satisfied, we know that $h$ satisfies the predicate $\pi_i$. This concludes this lemma.
\end{proof}

The previous result covers already a significant number of cases regarding the validity of the reduction. However, we might have that $h(z) = \#$ and $h(z') = 0$, switching the substitution of $z$ and $z'$, so to speak. Hence, before using the previous Lemmas to show the correctness of the reduction, we continue with that specific case. Indeed, in the case of $h(z) = \#$ and $h(z') = 0$, we can just consider encodings of computations of $A$ where the letters of the alphabet are swapped, using $0$'s to provide structural boundries and $\#$'s to encode the state and the counter values. Hence, for the rest of this section, assume $\ValCP(A)$ to be the set of encodings of valid computations of $A$ where the occurrences of the letters $0$ and $\#$ are swapped. We immediately obtain the following by Proposition \ref{prop:proof1sourceresult}.

\begin{corollary}\label{corollary:patlen-e-equiv-undec-old-predicates-permuted}
    Let $h\in H_{\ell_\alpha}$ with $h(z) = \#$, and $h(z') = 0$. Then, $h$ satisfies no predicate $\pi_1$ to $\pi_\mu$ if and only if $h(\alpha_1)\in\ValCP(A)$ and $h(\alpha_2)\in\{\#\}^*$ with $|h(\alpha_2)| > k$ for $k$ being the length of the three longest non-overlapping factors of $\#$'s in $h(\alpha_1)$ summed up.
\end{corollary}

Indeed, by Proposition \ref{prop:proof1sourceresult}, we know that this result already holds for $h(z) = 0$ and $h(z') = \#$ regarding $\ValC(A)$. As each predicate $\pi_i$, $i\in[\mu]$, is terminal-free and as the only terminals fixed are the ones given by $h(z)$ and $h(z')$, we know that switching this part of the substitution will switch the letters of the structural constraints/boundaries defined by each predicate, i.e., each position of all fixed $0$ and $\#$ occurrences. As no other terminals are used to obtain any structural property for $\pi_1$ to $\pi_\mu$, we know that we can obtain the same property of Proposition \ref{prop:proof1sourceresult} regarding $\ValCP(A)$ instead of $\ValC(A)$ if we swap each of those positions, i.e., by setting $h(z) = \#$ and $h(z') = 0$.

\textbf{Proving the reduction:} Finally, using all previously shown properties, we can prove the reduction in a concise manner.

$(\Rightarrow):$ We begin with the case that $\ValC(A) \neq \emptyset$. Then, a counter example to show $L_E(\alpha,\ell_\alpha)\neq L_E(\beta,\ell_\beta)$ suffices. Let $h\in H_{\ell_\alpha}$ such that $h(x_v) = 0\#\#\#0$, $h(z) = 0$, $h(z') = \#$, $h(y_1) = \varepsilon = h(y_2)$, $h(\alpha_1)\in\ValC(A)$, and $h(\alpha_2)\in\{0\}^*$ such that $|h(\alpha_2)|$ is longer than the three longest non-overlapping factors of $0$'s in $h(\alpha_1)$. By the definition of $\ell_\alpha$, we know that $h$ is $\ell_\alpha$-valid. Then, we obtain by Proposition \ref{prop:proof1sourceresult} that $h$ satisfies no predicate $\pi_1$ to $\pi_\mu$. Also, as $h(z) \neq h(z')$, we know by Lemma \ref{lemma:patlen-e-equiv-undec-equal-z-always-in-beta} that $h$ does not satisfy $\pi_{\mu+1}$. By the choice of $h$ and by Lemma \ref{lemma:patlen-e-equiv-undec-language-to-predicate}, we immediately get that $h(\alpha)\notin L_E(\beta,\ell_\beta)$, and hence $L_E(\alpha,\ell_\alpha) \neq L_E(\beta,\ell_\beta)$.

$(\Leftarrow):$ Now, assume $\ValC(A) = \emptyset$. By Proposition \ref{prop:proof1-langbeta-in-langalpha}, we know that $L_E(\beta,\ell_\beta) \subseteq L_E(\alpha,\ell_\alpha)$. So, we continue by showing $L_E(\alpha,\ell_\alpha) \subseteq L_E(\beta,\ell_\beta)$. Let $h\in H_{\ell_\alpha}$. We consider three cases, two of which follow with analogous arguments. First, if $h(z) = h(z')$, then we know by Lemma \ref{lemma:patlen-e-equiv-undec-equal-z-always-in-beta} that $h$ always satisfies $\pi_{\mu+1}$, hence, by Lemma \ref{lemma:palen-e-equiv-undec-predicate-language}, we know that $h(\alpha)\in L_E(\beta,\ell_\beta)$. Next, assume $h(z) = 0$ and $h(z') = \#$. As $\ValC(A) = \emptyset$, we know by Proposition \ref{prop:proof1sourceresult} that $h$ must satisfy some predicate $\pi_i$ for $i\in[\mu]$. Thus, by Lemma \ref{lemma:palen-e-equiv-undec-predicate-language}, we know that $h(\alpha)\in L_E(\beta,\ell_\beta)$. Now, assume $h(z) = \#$ and $h(z') = 0$. As $\ValC(A) = \emptyset$, we also know that $\ValCP(A) = \emptyset$ must be the case, as for each word in $\ValC(A)$, a corresponding word in $\ValCP(A)$ must exist and vice versa. Hence, by Corollary \ref{corollary:patlen-e-equiv-undec-old-predicates-permuted}, we know that $h$ must satisfy some predicate $\pi_i$ for $i\in[\mu]$. So, again, by Lemma \ref{lemma:palen-e-equiv-undec-predicate-language}, we know that $h(\alpha)\in L_E(\beta,\ell_\beta)$. This covers all cases, so we have $L_E(\alpha,\ell_\alpha) = L_E(\beta,\ell_\beta)$, concluding the binary case of the reduction.

\textbf{Larger alphabets: }To extend this proof to arbitrary larger alphabets $\Sigma$ which we fix w.l.o.g. by $\Sigma = \{0,\#,\ta_1,\ta_2,...,\ta_\sigma\}$ for $\sigma \geq 1$, we can create adapted patterns with length constraints $(\alpha_\sigma,\ell_{\alpha_\sigma})$ and $(\beta_\sigma,\ell_{\beta_\sigma})$ in the following manner, also similar to the construction in \cite{Freydenberger2010}. First, we define $\alpha_\sigma$ by
$$ \alpha_\sigma = y_1\ x_v\alpha_1x_v\alpha_2x_v\ y_2\ x_v\ z\ (z')^2\ (z_{\ta_1})^2\ (z_{\ta_2})^2\ ...\ (z_{\ta_\sigma})^2\ z\ x_v .$$
Notice, that only the parts containing the $z$'s is changed. We construct the adapted length constraints $\ell_{\alpha_\sigma}$ in the following manner:
\begin{itemize}
    \item $x_v = 5$
    \item $z = 1$
    \item $z' = 1$
    \item $z_{\ta_i} = 1$ for all $i\in[\sigma]$
\end{itemize}
An immediate observation that can be made, as its done in \cite{Freydenberger2010}, is the following:

\begin{observation}\cite{Freydenberger2010}
    Let $n \geq 1$, $x_1,...,x_n\in X$, and $\ta_1,\ta_2,...,\ta_n\in\Sigma$. Consider the pattern
    $$ p = x_1\ x_2x_2\ x_3x_3\ ...\ x_nx_n\ x_1.$$
    Then, for each homomorphism $h$ with $h(p) = \ta_1\ \ta_2\ta_2\ \ta_3\ta_3\ ...\ \ta_n\ta_n\ \ta_1$ we have that $h(x_i) = \ta_i$ for all $i\in[n]$.
\end{observation}

Now, we explain the changes necessary to create $\beta_\sigma$ and $\ell_{\beta_\sigma}$. First, for each $\eta_i$ with $i\in[\mu]$, we change its form from $\eta_i = z_iz_i'z_i'z_i$ to
$$\eta_i = z_i\ z_i'z_i'\ z_{i,\ta_1}z_{i,\ta_1}\ ...\ z_{i,\ta_\sigma}z_{i,\ta_\sigma}\ z_i.$$
Now, instead of adding $\eta_{\mu+1}$ as before, we add a series of $n\in\N$ many predicates that cover the cases that two different $z$ variables are substituted equally in $\alpha_\sigma$. In each of these newly defined predicates $\pi_i$, we would set $\gamma_i$ and $\delta_i$ just so single, new, and independent variables to obtain any substitution $h(\alpha_1)$ and $h(\alpha_2)$. The number of those additional predicates becomes significantly for alphabets of larger size, hence we only give an example here. 

Consider the case that $h(z') = h(z_{\ta_2})$. To handle this case, we would add a predicate $\pi_{\mu+k}$ for some $k\in[n]$ such that
$$ \eta_{\mu+k} = z_{\mu+k}\ (z_{\mu+k}')^2\ (z_{\mu+k,\ta_1})^2\ (z_{\mu+k}')^2\ (z_{\mu+k,\ta_3})^2\ ...\ (z_{\mu+k,\ta_\sigma})^2\ z_{\mu+k}$$
and $\gamma_{\mu+k} = y_{\mu+k}$ as well as $\delta_{\mu+k} = y_{\mu+k}'$ where each new variable is new and independent from other parts of the pattern. In particular, notice, that in this example we have no occurrence of the variable $z_{\mu+k,\ta_2}$ and instead four occurrences of $z_{\mu+k}'$. The idea is similar to the previous construction of $\eta_{\mu+1}$ in the binary case to handle the case of $h(z) = h(z')$.

We observe, that if for some substitution $h\in H_{\ell_{\alpha_\sigma}}$ any two different variables with the base name $z$ are substituted equally, then there exists one predicate between $\pi_{\mu+1}$ to $\pi_{\mu+n}$ that has the same variables equalled out, hence resulting in that predicate being satisfied, as $\gamma$ and $\delta$ are simply just independent variables in that case.

Additionally, we add $\sigma$ many predicates $\pi_{\mu+n+1}$ to $\pi_{\mu+n+\sigma}$ that handle the occurrence of any of the letters $h(z_{\ta_1})$ to $h(z_{\ta_\sigma})$ in $h(\alpha_1)$ or $h(\alpha_2)$. These predicates work exactly as in \cite{Freydenberger2010}, just with the same addition as before that $h(z)$ and $h(z')$ determine the letters that need to be used for the encodings of valid computations and for the substitutions ob $x_v$. All other arguments stay the same. Hence, we also conclude for larger alphabets that the equivalence problem for erasing pattern languages with length constraints is undecidable. This concludes the proof of Theorem \ref{theorem:patlang-len-erasing-equivalence-undecidable}.
\newpage

\section{The Undecidability of the Inclusion Problem for Terminal-Free Non-Erasing Pattern Languages with Length Constraints}
\label{section:proof2}
This section is dedicated to show Theorem \ref{theorem:patlang-len-nonerasing-inclusion-tf-undecidability}. The proof is based on the construction of the proof of undecidability of the inclusion problem of general non-erasing pattern languages without any constraints done by Bremer and Freydenberger in \cite{BREMER201215} and uses slight adaptations and length constraints to allow for this construction to work in the terminal-free case. The undecidability of the binary case immdediately follows by the result from Saarela \cite{saarela:LIPIcs.ICALP.2020.140}. To aid understanding the proof for general alphabets of size greater or equal to $3$ for patterns with length constraints, we begin by providing an alternative approach for such patterns based on the construction in \cite{BREMER201215} for binary alphabets and extend to general alphabet sizes after that.

We begin with the binary case and assume that $\Sigma = \{0,\#\}$. Similar to \cite{BREMER201215}, we reduce the emptiness problem of a specific universal Turing machine $U$ to the inclusion problem of non-erasing pattern languages, but restricted to terminal-free patterns. This section shows that for two given terminal free patterns with length constraints $(\alpha,\ell_\alpha)$ and $(\beta,\ell_\beta)$ we have $\ValC(U) = \emptyset$ if and only if $L_{NE}(\alpha,\ell_\alpha) \subseteq L_{NE}(\beta,\ell_\beta)$.

We begin by giving the adapted versions of the patterns and start with the pattern $(\beta,\ell_\beta)$ first, as the construction of $(\alpha,\ell_\alpha)$ is based on it. We define $\beta$ by
$$ \beta = x_0\ x_a x_b\ x_\#^5\ x_ax_1...x_{\mu+1}x_b\ x_\#^5\ r_1\hat{\beta}_1r_2\hat{\beta}_2\ ...\ r_{\mu+1}\hat{\beta}_{\mu+1}r_{\mu+2} $$
for new and independent variables $x_0, x_\#, x_a, x_b, x_1, x_2,...,x_{\mu+1},r_1,r_2,...,r_{\mu+2}\in X$ and terminal-free patterns $\hat{\beta}_1,\hat{\beta_2},...,\hat{\beta}_{\mu+1}\in X^+$. Notice that $\mu\in\N$ is a number given by the construction in \cite{BREMER201215}. For all $i\in[\mu+1]$ we say that
$$ \hat{\beta}_i = x_0x_i^4x_0\ \gamma_i\ x_0x_i^4x_0\ \delta_i\ x_0x_i^4x_0$$
for terminal-free patterns $\gamma_i,\delta_i\in X^+$ that are defined later. The length constraints $\ell_\beta$ are defined by the system
\begin{itemize}
    \item $x_0 = 1$
    \item $x_\# = 1$
    \item $x_i = 1$ for all $i\in[\mu+1]$
    \item $x_a + x_b = \mu+2$
\end{itemize}
Let $\tau : (\Sigma\times X)^* \rightarrow X^*$ be a homomorphism defined by $\tau(0) = x_0$, $\tau(\#) = x_\#$, and $\tau(x) = x$ for all $x\in X$. Then we say that for all $i\in[\mu]$ we have $\gamma_i = \tau(\gamma_i')$ and $\delta_i = \tau(\delta_i')$ for $\gamma_i'$ and $\delta_i'$ being the patterns used in the construction in \cite{BREMER201215}. The specific construction of each such pattern is not relevant here and for details we refer to \cite{BREMER201215}. What is important is that we use the same patterns as in \cite{BREMER201215} but map them to terminal free patterns that have the variable $x_0$ instead of each occurrence of a terminal symbol $0$ and the variable $x_\#$ instead of each occurrence of $\#$. For $\mu+1$ we define
\begin{align*}
    \gamma_{\mu+1} &= x_0^{|I|+5} y_{\mu+1} \\
    \delta_{\mu+1} &= y_{\mu+1}'
\end{align*}
for new and independent variables $y_{\mu+1}, y_{\mu+1}'\in X$ and $I$ being the encoding of the initial configuration of $U$ as in \cite{BREMER201215}. This will be used to obtain all substitutions $h(\alpha)$ where we have $h(x_0) = h(x_\#)$ in $L_{NE}(\beta,\ell_\beta)$.

Now, we construct $(\alpha,\ell_\alpha)$. $\alpha$ is defined by
$$ \alpha = x_0^{\mu+3}\ x_\#^5\ x_0^{\mu+1}x_\#x_0^{\mu+1}\ x_\#^5\ tv\ x_0\alpha_1x_0\ v\ x_0\alpha_2x_0\ vt$$
for some patterns $\alpha_1$ and $\alpha_2$, $x_0$ and $x_\#$ defined as in $\beta$, $v = x_0x_\#x_\#x_\#x_0$ and $t = \psi(r_1\hat{\beta}_1r_2...r_{\mu+1}\hat{\beta}_{\mu+1}r_{\mu+2})$ with $\psi : X^* \rightarrow  X^*$ defined by $\psi(x_0) = x_0$, $\psi(x_\#) = x_\#$, and $\psi(x) = x_0$ for all other $x\in\var(\beta)$. We define $\alpha_1$ and $\alpha_2$ as $\alpha_1 = \tau(\alpha_1')$ and $\alpha_2 = \tau(\alpha_2')$ where $\alpha_1'$ and $\alpha_2'$ are the patterns given in \cite[Section 4.4]{BREMER201215} for the second case. 

The most important fact we take from there is that $\alpha_1' = \#\#I\#\#x\#0^60^10\#\#$ for some new and independent variable $x\in X$. By that, we see that $\tau(\alpha_1')$ starts with $x_\#x_\#\tau(I)x_\#x_\#$ which is used later on in this adaptation of the proof. The length constraints $\ell_\alpha$ are defined by the equations $x_0 = 1$ and $x_\# = 1$. Notice that the $x$ in $\alpha_1$ has no length constraint. With that, we can now show the correctness of the reduction for the binary case.

First, assume that $\ValC(U) \neq \emptyset$. Let $h\in H_{\ell_\alpha}$ be some substitution such that $h(x_0) = 0$, $h(x_\#) = \#$.
We know by the construction of $(\beta,\ell_\beta)$ that for some substitution $h'\in H_{\ell_\beta}$ we can have $h(\alpha) = h'(\beta)$ if and only if $h'(x_0) = h(x_0) = 0$ and $h'(x_\#) = h(x_\#) = \#$ by the fact that $\alpha$ starts with $x_0^{\mu+3}x_\#^5$ and $\beta$ starts with $x_0x_ax_bx_\#^5$ and $\ell_\beta$ as well as $\ell_\alpha$ require $h(x_0^{\mu+3}x_\#^5) = 0^{\mu+3}\#^5 = h'(x_0x_ax_bx_\#^5)$ with $|h(x_0)| = 1 = |h'{x_0}|$ and $|h'(x_ax_b)| = \mu+2$. Now, we know by \cite{BREMER201215} that $h(\alpha)\in L_{NE}(\beta,\ell_\beta)$ if and only if there exists some $i\in[\mu+1]$ such that $\tau'(\gamma_i) = h(0\alpha_10)$ and $\tau'(\delta_i) = h(0\alpha_20)$ for some morphism $\tau'$ with, in our case, $\tau'(x_0) = 0$ and $\tau'(x_\#) = \#$. For the case of $\gamma_{\mu+1}$ and $\delta_{\mu+1}$ we know that we cannot use those as $\tau'(\gamma_{\mu+1})$ always starts with $|I|+5$ many $0's$ but $h(\alpha)$ contains occurrences of $\#$, because $h(\alpha)$ starts with $h(x_\#x_\#I\#\#) = \#\#I\#\#$. So, we can restrict ourselves to $i\in[\mu]$. Then, by the previous restrictions and the choice of $\hat{\beta}_1$ to $\hat{\beta}_\mu$, we know by \cite{BREMER201215} that there exists some substitution $h'\in H_{\ell_\beta}$ with $h'(\beta) = h(\alpha)$ if and only if $h(\alpha_1)\notin \ValC(U)$ or $h(\alpha_2)\in\{0\}^*$ such that $|h(\alpha_2)|$ is short enough. Hence, as $\ValC(U) \neq \emptyset$, we can choose $h(\alpha_1)\in V$ and $h(\alpha_2)\in\{0\}^*$ long enough such that $h(\alpha,\ell_\alpha)\notin L_{NE}(\beta,\ell_\beta)$ and by that $L_{NE}(\alpha,\ell_\alpha)\not\subset L_{NE}(\beta,\ell_\beta)$.

Now, for the other direction, assume $\ValC(U) = \emptyset$. Let $h\in H_{\ell_\alpha}$. We consider three cases. First, assume $h(x_0) = h(x_\#)$. Then, we know that $h(\alpha_1)$ starts with $h(x_0)^{|I|+5}$. So, we can select $h'$ such that $h'(\gamma_{\mu+1}) = h(0\alpha_10)$, $h'(\delta_{\mu+1}) = h(0\alpha_20)$, $h'(r_{\mu+1}) = h(\psi(r_{\mu+1}\hat{\beta}_{\mu+1}r_{\mu+2}))$, $h'(r_{\mu+2}) = h(\psi(r_1\hat{\beta}_1...\hat{\beta}_{\mu+1}r_{\mu+2}))$, and $h'(x) = h(x_0)$ for any other variable $x\in\var(\beta)$. Notice that we then have $h'(r_1\hat{\beta}_1...r_{\mu+1}) = h'(\psi(r_1\hat{\beta}_1...\hat{\beta}_{\mu+1}r_{\mu+2})) = h'(t) = h(t)$, $h'(r_{\mu+2}) = h'(t) = h(t)$. We obtain 
\begin{align*}
    h(\alpha) &= h(x_0^{\mu+3}\ x_\#^5\ x_0^{\mu+1}x_\#x_0^{\mu+1}\ x_\#^5\ tv\ x_0\alpha_1x_0\ v\ x_0\alpha_2x_0\ vt) \\
              &= h(x_0^{3\mu+15})h(t)h(v)h(x_0\alpha_1x_0)h(v)h(x_0\alpha_2x_0)h(v)h(t) \\
              &= h'(x_0^{3\mu+15})h(t)h'(x_0x_\#^4x_\#)h'(\gamma_{\mu+1})h'(x_0x_\#^4x_\#)h(\delta_{\mu+1})h'(x_0x_\#^4x_\#)h(t) \\
              &= h'(x_0^{3\mu+15})h(t)h'(\hat{\beta}_{\mu+1})h(t) \\
              &= h'(x_0x_ax_bx_\#^5x_ax_1...x_{\mu+1}x_bx_\#^5)h'(r_1\hat{\beta}_1...r_{\mu+1})h'(\hat{\beta}_{\mu+1})h'(r_{\mu+2}) \\
              &= h'(\beta)
\end{align*}
and by that $h(\alpha) \in L_{NE}(\beta,\ell_\beta)$. 
For the second case, assume $h(x_0) = 0$ and $h(x_\#) = \#$. Then, notice the argumentation from the other direction of this proof. As $\ValC(U) = \emptyset$, we know by \cite{BREMER201215} that there always exists some $i\in[\mu]$ and $h'\in H_{\ell_\beta}$ such that $h'(\gamma_i) = h(0\alpha_10)$ and $h'(\delta_i) = h(0\alpha_20)$ if we set $h'(x_0) = 0$, $h'(x_\#) = \#$. By setting $h'(x_i) = \#$, $h'(x_a) =  0^{\mu-i+2}$,$h'(x_\beta) = 0^{i}$ , $h'(r_i) = h(\psi(r_i\hat{\beta}_i...r_{\mu+1}))$, $h'(r_{i+1}) = h(\psi(r_1\hat{\beta}_1...r_{i+1}))$, and $h'(x) = 0$ for all other $x\in\var(\beta)$, we obtain
\begin{align*}
    h(\alpha) &= h(x_0^{\mu+3}\ x_\#^5\ x_0^{\mu+1}x_\#x_0^{\mu+1}\ x_\#^5\ tv\ x_0\alpha_1x_0\ v\ x_0\alpha_2x_0\ vt) \\
              &= 00^{\mu+2}\ \#^5\ 0^{\mu+1}\#0^{\mu+1}\ \#^5\ h(t)\ 0\#^40\ h(0\alpha_10)\ 0\#^40\ h(0\alpha_20)\ 0\#^40\ h(t) \\
              &= 00^{\mu-i+2}0^{i}\ h'(x_\#^5)\ 0^{\mu-i+2}0^{i-1}\#0^{\mu+1-i}0^{i}\ h'(x_\#^5)\ h(t) h'(\hat{\beta}_i) h(t) \\
              &= h'(x_0x_ax_b\ x_\#^5\ x_ax_1...x_{i-1}x_ix_{i+1}...x_{\mu+1}x_b\ x_\#^5)\ h(t) h'(\hat{\beta}_i) h(t) \\
              &= h'(x_0x_ax_bx_\#^5x_ax_1...x_{i-1}x_ix_{i+1}...x_{\mu+1}x_bx_\#^5)h'(r_1\hat{\beta}_1...r_{i}) h'(\hat{\beta}_i) h'(r_{i+1}\hat{\beta}_{i+1}...r_{\mu+2}) \\
              &= h'(\beta)
\end{align*}
as $h(t) = h(\psi(r_1\hat{\beta}_1...r_{\mu+2})) = h'(r_1\hat{\beta}_1...r_{i}) = h'(r_{i+1}\hat{\beta}_{i+1}...r_{\mu+2})$ by the definition of $h'$. By that, in this case, we also have $h(\alpha)\in L_{NE}(\beta,\ell_\beta)$.
Now, for the final case, we notice that if $h(x_0) = \#$ and $h(x_\#) = 0$, then all previously proven properties for the case of $h(x_0) = 0$ and $h(x_\#) = \#$ work exactly in the same way, just with the letters $0$ and $\#$ flipped. As $\ValC(U) = \emptyset$, there cannot exist any encoding of a valid computation, even with occurrences of $0$'s and $\#$'s changed. Hence, in this case, we also have $h(\alpha)\in L_{NE}(\beta,\ell_\beta)$ by an analogous argument as in the second case. This concludes all cases of this direction, thus $L_{NE}(\alpha,\ell_\alpha) \subseteq L_{NE}(\beta,\ell_\beta)$.

\textbf{Larger Alphabets:} Now, for larger alphabets, assume w.l.o.g.  $\Sigma = \{0,\#,\ta_1,...,\ta_\sigma\}$ for some $\sigma\in\N$. In contrast to the construction in \cite{BREMER201215}, the adaptations we have to make to the original patterns $(\alpha,\ell_\alpha)$ and $(\beta,\ell_\beta)$ are now more substantial. The basic idea of the construction stays the same. For each new $\ta_i$, we introduce a new variable $x_{\ta_i}$. We adapt the patterns in such a way that for each valid substitution $h\in H_{\ell_\alpha}$ that maps two of those variables $x_{\ta_i}$ and $x_{\ta_j}$, representing single letters, to the same letter, i.e., $h(x_{\ta_i}) = h(x_{\ta_j})$, then $h(\alpha)$ is always in $L_{NE}(\beta,\ell_\beta)$. In addition to that, it is ensured that if any of the letters obtained by $h(x_{\ta_1})$ to $h(s_{\ta_\sigma})$ occur in $h(\alpha_1)$ or $h(\alpha_2)$, then $h(\alpha)$ can also always be found in $L_{NE}(\beta,\ell_\beta)$.

What follows is the adapted constructions of the two patterns with length constraints $(\alpha',\ell_{\alpha'})$ and $(\beta',\ell_{\beta'})$ and then the argumentation why the reduction also works for those. First, we redefine each of the patterns $\hat{\beta}_1$ to $\hat{\beta}_{\mu+1}$ by redefining the $\gamma_i$'s in those. For each $i\in[\mu+1]$ we redefine $\hat{\beta}_i$ by
$$ \hat{\beta}_i = x_0x_\#^4x_0\ \mathbf{x_0x_\#x_{\ta_1}x_{\ta_2}...x_{\ta_\sigma}\gamma_i}\ x_0x_\#^4x_0\ \delta_i\ x_0x_\#^4x_0$$
for new and independent variables $x_{\ta_1}$ to $x_{\ta_\sigma}$ (the letters in bold font represent the new $\gamma_i'$). Hence, for each $\gamma_i$ we just added the prefix $x_0x_\#x_{\ta_1}...x_{\ta_\sigma}$, but left everything else the same. Additionally, instead of having $\mu+1$ many patterns $\hat{\beta}_i$, we will now have $\mu+1+2\sigma+n$ for some $n\in\N$ to be defined later, hence resulting in the variables $r_i$ and $x_i$ to range between $r_1$ to $r_{\mu+1+2\sigma+n+1}$ (respectively $x_1$ to $x_{\mu+1+2\sigma+n})$ instead of $r_1$ to $r_{\mu+1}$ (respectively $x_1$ to $x_{\mu+1}$). The rest of the pattern $\beta$ is left untouched. Now, for this adapted definition of $\beta$, we can now define $\beta'$ by
$$ \beta' = x_0x_\#x_{\ta_1}x_{\ta_2}...x_{\ta_\sigma}\beta$$
while the variables in the front are the same that have been prepended to all patterns $\hat{\beta}_1$ to $\hat{\beta}_{\mu+1}$. The length constraints $\ell_{\beta'}$ are almost the same as the ones in $\ell_\beta$ but with one additional constraint and a change for the sum of $x_a$ and $x_b$:
\begin{itemize}
    \item $x_0 = 1$ 
    \item $x_\# = 1$
    \item $x_a + x_b = \mu+1+2\sigma+n+1$
    \item $x_i = 1$ for all $i\in[\mu+1+2\sigma+n]$
    \item $x_{\ta_i} = 1$ for all $i\in[\sigma]$
\end{itemize}
So, in total, the pattern with length constraints $(\beta',\ell_{\beta'})$ looks almost identical to $\beta$ but with the addition of the variables in the front, the variables in the beginning of each $\gamma_i'$ for $i\in[\mu+1]$, and the addition of more patterns $\hat{\beta}_j$ to be defined later.

To define $\alpha'$, we first have to extend the morphism $\psi$ to obtain adapted versions of the patterns $t$ as before. The extension $\psi' : X^* \rightarrow X^*$ is defined by $\psi'(x_{\ta_i}) = x_{\ta_i}$ for all $i\in[\sigma]$ and for all other cases $\psi'(x) = \psi(x)$ for all other variables $x\in X$. Additionally, we need adapted versions of the patterns $\alpha_1$ and $\alpha_2$ that occur in $\alpha$. We construct the adapted versions $\alpha_1'$ and $\alpha_2'$ by $\alpha_1' = x_0x_\#x_{\ta_1}x_{\ta_2}...x_{\ta_\sigma}\alpha_1$ and $\alpha_2' = x_0x_\#x_{\ta_1}x_{\ta_2}...x_{\ta_\sigma}\alpha_2$. Now, with $t = \psi'(r_1\hat{\beta}_1r_2\ ...\ r_{\mu+1+2\sigma+n+1}\hat{\beta}_{\mu+1+2\sigma+n+1}r_{\mu+1+2\sigma+n+2})$ and assuming that $\alpha$ now contains the adapted $\alpha_1'$ and $\alpha_2'$ instead of $\alpha_1$ and $\alpha_2$, we define $\alpha'$ by
$$ \alpha' = x_0x_\#x_{\ta_1}x_{\ta_2}...x_{\ta_\sigma}\alpha.$$
Notice, that the only difference to the pattern of the binary case is that in this case also just pattern $x_0x_\#x_{\ta_1}x_{\ta_2}...x_{\ta_\sigma}$ has been prepended to $\alpha'$, $\alpha_1$, and $\alpha_2$.

The following idea, which is to be shown, works as follows: First, we define $2\sigma$ additional patterns $\hat{\beta}_i$ for $i\in[\mu+1+2\sigma]\setminus[\mu+1]$ that capture the case that all variables $x_{\ta_1}$ to $x_{\ta_\sigma}$ as well as $x_0$ and $x_\#$ are substituted to different letters, but one of the letters obtained by $x_{\ta_1}$ to $x_{\ta_\sigma}$ occurs later on in any substitution of $\alpha_1$ or $\alpha_2$. For all other cases that any two of the variables $x_{\ta_1}$ to $x_{\ta_\sigma}$ as well as $x_0$ and $x_\#$ are substituted equally, similar to the proof of Theorem \ref{theorem:patlang-len-erasing-equivalence-undecidable}, we add $n$ additional patterns $\hat{\beta}_i$ for $i\in[\mu+1+2\sigma+n]\setminus[\mu+1+2\sigma]$ that can be used as wildcards for any substitution of $\alpha_1'$ or $\alpha_2'$. By that, if we substitute each of the variables $x_{\ta_1}$ to $x_{\ta_\sigma}$ as well as $x_0$ and $x_\#$ with different letters, we can still only use the patterns $\hat{\beta}_1$ to $\hat{\beta}_{\mu+1}$ and use the same results as before to obtain the reduction.

We begin with the construction of $\hat{\beta}_{\mu+2}$ to $\hat{\beta}_{\mu+2+2\sigma}$. For each $i\in[\sigma]$, we define $\gamma_{\mu+1+i}$ and $\delta_{\mu+1+i}$ as well as $\gamma_{\mu+1+\sigma+i}$ and $\delta_{\mu+1+\sigma+i}$ by
\begin{align*}
    \gamma_{\mu+1+i} &= x_0\ x_0x_\#x_{\ta_1}...x_{\ta_\sigma}\ y_{\mu+1+i,1}x_{\ta_i}y_{\mu+1+i,2}, \\
    \delta_{\mu+1+i} &= y_{\mu+1+i}', \\
    \gamma_{\mu+1+\sigma+i} &= y_{\mu+1+\sigma+i}, \text{ and} \\
    \delta_{\mu+1+\sigma+i} &= x_0\ x_0x_\#x_{\ta_1}...x_{\ta_\sigma}\ y_{\mu+1+\sigma+i,1}'x_{\ta_i}y_{\mu+1+\sigma+i}'
\end{align*}
for new and independent variables $y_{\mu+1+i,1}$, $y_{\mu+1+i,2}$, $y_{\mu+1+i}'$, $y_{\mu+1+\sigma+i}$, $y_{\mu+1+\sigma+i,1}'$, and $y_{\mu+1+\sigma+i,2}'$. We obtain the first property mentioned above.

\begin{lemma}\label{lemma:patlang-len-ne-tf-inclusion-large-alphabet-more-predicates-1}
    Let $h\in H_{\ell_{\alpha'}}$. If for all two variables $x,y\in\{x_0,x_\#,x_{\ta_1},...,x_{\ta_\sigma}\}$ with $x \neq y$ we have $h(x) \neq h(y)$, and we have for some $x_{\ta_i}$ with $i\in[\sigma]$ that $h(\alpha_1') = h(x_0x_\#x_{\ta_1}...x_{\ta_\sigma})uh(x_{\ta_i})v$ or $h(\alpha_2') = h(x_0x_\#x_{\ta_1}...x_{\ta_\sigma})uh(x_{\ta_i})v$ for some $u,v\in\Sigma^*$, then $h(\alpha')\in L_E(\beta',\ell_{\beta'})$.
\end{lemma}
\begin{proof}
    Assume all variables introduced as in the lemma. We construct $h'$ by the following. For all $x\in\{x_0,x_\#,x_{\ta_1},...,x_{\ta_\sigma}\}$ we set $h'(x) = h(x)$. Assume $h(\alpha_1') = h(x_0x_\#x_{\ta_1}...x_{\ta_\sigma})uh(x_{\ta_i})v$ (the same following arguments also holds for $\alpha_2'$ analogously by choosing the corresponding $\hat{\beta}_{\mu+1+\sigma+i}$). By the construction of $\gamma_{\mu+1+i}$ we know that we can set $h'(y_{\mu+1+i,1}) = u$ and $h'(y_{\mu+1+i,2}) = v$ to obtain $h'(\gamma_{\mu+1+i}) = h(x_0\alpha_1'x_0)$. Additionally, we can immediately set $h'(\delta_{\mu+1+i}) = h'(y_{\mu+1+i}') = h(x_0\alpha_2x_0)$. Now, set $h'(x_i) = h(x_\#)$, $h'(x_a) = h(x_0)^{2\sigma+n-i+1}$, $h'(x_b) = h(x_0)^{\mu+1+i}$, $h'(r_{\mu+1+i}) = h(\psi'(r_{\mu+1+i}\hat{\beta}_{\mu+1+i}...r_{\mu+1+2\sigma+n}))$ and $h'(r_{\mu+1+i+1}) = h(\psi(r_1\hat{\beta}_1...r_{\mu+1+i+1}))$. For all other variables $x$ we set $h'(x) = h(x_0)$. Now, using analogous derivations as in the second direction of the proof for the reduction in the binary case, we obtain that $h'(\beta') = h(\alpha')$ and, by that, as $h'$ is $\ell_{\beta'}$-valid, $h(\alpha') \in L_{NE}(\beta,\ell_\beta)$.
\end{proof}

Hence, we only have to construct the final $n$ patterns that cover the cases of two of the variables $x_{\ta_1}$ to $x_{\ta_\sigma}$ as well as $x_0$ and $x_\#$ being substituted equally. For that, we give a general outline of the construction of the patterns $\hat{\beta}_{\mu+1+2\sigma+1}$ to $\hat{\beta}_{\mu+1+2\sigma+n}$. Let $x,y\in\{x_0x_\#x_{\ta_1}...,x_{\ta_\sigma}\}$ with $x \neq y$. W.l.o.g. assume $x = x_{\ta_i}$ and $y = x_{\ta_j}$ for some $i < j$. For the cases of $x = x_0$ or $y = x_\#$, the same construction can be done. Also, w.l.o.g. assume this combination of variables is handled in $\hat{\beta}_{\mu+1+2\sigma+k}$ for some $k\in[n]$. Then, we define $\gamma_{\mu+1+2\sigma+k}$ and $\delta_{\mu+1+2\sigma+k}$ by 
\begin{align*}
    \gamma_{\mu+1+2\sigma+k} &= x_0\ x_0x_\#x_{\ta_1}...x_{\ta_i}...x_{\ta_{j-1}}x_{\ta_i}x_{\ta_{j+1}}...x_{\ta_\sigma}y_{\mu+1+2\sigma+k}, \text{ and}\\
    \delta_{\mu+1+2\sigma+k} &= x_0\ x_0x_\#x_{\ta_1}...x_{\ta_i}...x_{\ta_{j-1}}x_{\ta_i}x_{\ta_{j+1}}...x_{\ta_\sigma}y_{\mu+1+2\sigma+k}'
\end{align*}
for new and independent variables $y_{\mu+1+2\sigma+k}$ and $y_{\mu+1+2\sigma+k}'$. Notice that in these patterns $x_{\ta_j}$ does not occur and is replaced by another occurrence of $x_{\ta_i}$. For all other $k'\in[n]$ we construct all other possible combinations of 2 different variables in that way analogously. We obtain the following final result.

\begin{lemma}
    Let $h\in H_{\ell_{\alpha'}}$. If there exist two variables $x,y\in\{x_0,x_\#,x_{\ta_1},...,x_{\ta_\sigma}\}$ with $x \neq y$ such that $h(x) = h(y)$, then we have $h(\alpha')\in L_{NE}(\beta',\ell_\beta')$.
\end{lemma}
\begin{proof}
    Suppose w.l.o.g. that $h(x_{\ta_i}) = h(x_{\ta_j})$ for $i,j\in[\sigma]$ and $i < j$. The same approach works if any of the variable was $x_0$ or $x_\#$ as we assumed the previous construction to also be done for those variables as well. We know there exists some $k\in[n]$ such that $\gamma_{\mu+1+2\sigma+k} = x_0x_0x_\#x_{\ta_1}...x_{\ta_i}...x_{\ta_{j-1}}x_{\ta_i}x_{\ta_{j+1}}...x_{\ta_\sigma}y_{\mu+1+2\sigma+k}$ and $\delta_{\mu+1+2\sigma+k} = x_0x_0x_\#x_{\ta_1}...x_{\ta_i}...x_{\ta_{j-1}}x_{\ta_i}x_{\ta_{j+1}}...x_{\ta_\sigma}y_{\mu+1+2\sigma+k}'$. By the construction of $\alpha_1'$ and $\alpha_2'$ we know that both have the prefix $h(x_0x_\#x_{\ta_1}...x_{\ta_\sigma})$. So, as before, we can use $\gamma_{\mu+1+2\sigma+k}$ and $\delta_{\mu+1+2\sigma+k}$ to define some $h'\in H_{\ell_{\beta'}}$ with $h'(\gamma_{\mu+1+2\sigma+k}) = h(x_0\alpha_1'x_0)$ and $h'(\delta_{\mu+1+2\sigma+k}) = h(x_0\alpha_2'x_0)$. Now, by setting $h'(x_{\mu+1+2\sigma+k}) = h(x_\#)$ and everything else analogously to, e.g., Lemma \ref{lemma:patlang-len-ne-tf-inclusion-large-alphabet-more-predicates-1}, we obtain that we can find such a $h'$ with $h'(\beta') = h(\alpha')$ and by that $h(\alpha')\in L_{NE}(\beta',\ell_{\beta'})$.
\end{proof}

With all these properties, we know that for some $h\in H_{\alpha',\ell_{\alpha'}}$ to result in a word such that $h(\alpha')\notin L_{NE}(\beta',\ell_{\beta})$, we need that for all variables $x,y\in\{x_0,x_\#,x_{\ta_1},...,x_{\ta_\sigma}\}$ with $x \neq y$ we have $h(x) \neq h(y)$. Notice, that in this case, the general structure of $\beta'$ and $\alpha'$ are not very different from the binary case, except a constant length prefix with $\sigma+2$ different letters. By the arguments from the binary base, we see that the existence of some $\gamma_i$ and $\delta_i$ to obtain $h(x_0\alpha_1'x_0)$ and $h(x_0\alpha_2'x_0)$ is still a necessary condition to find some $h'\in H_{\ell_{\beta'}}$ for which we obtain $h'(\beta') = h(\alpha')$. Notice that none of the $n$ last $\hat{\beta}_{i}$ can be applied in these cases, as no two variables in $\{x_0,x_\#,x_{\ta_1},...,x_{\ta_\sigma}\}$ are substituted equally. Also notice, by Lemma \ref{lemma:patlang-len-ne-tf-inclusion-large-alphabet-more-predicates-1}, if any of the letters $h(x_{\ta_1})$ to $h(x_{\ta_\sigma})$ appear somewhere later on in $h(\alpha_1')$ or $h(\alpha_2')$, that $h(\alpha')$ is immediately in $L_{NE}(\beta,\ell_{\beta'})$. Hence, only binary words over $h(x_0)$ and $h(x_\#)$ obtained by $h(\alpha_1')$ and $h(\alpha_2')$ could result in substitutions $h(\alpha')$ that are not in $L_{NE}(\beta,\ell_{\beta'})$. Hence, only one of the initial patterns $\hat{\beta}_1$ to $\hat{\beta}_\mu$ may be applied to find some substitution $h'\in H_{\ell_{\beta'}}$ to get $h(\alpha') = h'(\beta')$. Hence, the reduction follows by the arguments done in the binary case for these initial $\mu$ predicates.

This concludes the case for all larger alphabets $\Sigma$ with $|\Sigma| \geq 3$ and by that we know that the inclusion problem of terminal-free pattern languages with length constraints is undecidable for all alphabets $\Sigma$ with $|\Sigma| \geq 2$.
\newpage

\section{The Undecidability of the Equivalence Problem for Non-Erasing Pattern Languages with Regular- and Length Constraints}
\label{section:proof3}
This section is dedicated to show Theorem \ref{theorem:patlang-reglen-nonerasing-equivalence-undecidability}. For this result, we also reduce the emptiness problem of nondeterministic 2-counter automata without input to the problem of deciding whether two pattern languages with regular and length constraints are equal. In contrast to the other two main results, this result uses a disjunction of systems of linear (diophantine) inequalities instead of just a single system as the considered length constraints. This makes this result somewhat weaker than the other two, but nonetheless serves as a first upper bound of undecidability for a problem which is so far only known to be trivially decidable in the case of no constraints. Notice, that this proof is partly based on ideas from \cite{Nowotka2024} and \cite{BREMER201215} in the way the patterns are structured and the regular constraints are obtained. Nontheless, certain significant changes were necessary to obtain the result.

We begin with the binary case and assume w.l.o.g. $\Sigma = \{0,\#\}$. Let $A = (Q,\delta,q_0,F)$ be some nondeterministic 2-counter automaton without input. For this proof, we assume that for any configuration $C_i = (q_j,m_1,m_2)$ of $A$ we have that its encoding is defined by $\enc(C_i) = 0^{1+j}\#0^{1+m_1}\#0^{1+m_2}$ for a state $q_j\in Q$ and counter values $m_1,m_2\in\N$. The encoding of a computation $C = (C_1,C_2,...,C_n)$ for some $n\in\N$ with $n \geq 1$ is still defined as $\enc(C) = \#\# \enc(C_1) \#\# \enc(C_2) \#\# ... \#\# \enc(C_n) \#\#$. Also, for this section, assume that, given some regular expression $R$, we have that $L(R)$ refers to the regular language obtained by that regular expression.

For the reduction, we will construct two patterns with regular and length constraints $(\alpha,r_\alpha,\ell_\alpha)$ and $(\beta,r_\beta,\ell_\beta)$ for which we have $L_{NE}(\alpha,r_\alpha,\ell_\alpha) = L_{NE}(\beta,r_\beta,\ell_\beta)$ if and only if $\ValC(A) = \emptyset$.

As we need $\beta$ to define $\alpha$, we will begin with the construction of $(\beta,r_\beta,\ell_\beta)$. We define $\beta$ by
$$ \beta = r_1\ \hat{\beta}_1\ r_2\ ... r_\mu\ \hat{\beta}_\mu\ r_{\mu+1}$$
for new and independent variables $r_1$ to $r_{\mu+1}$ for some $\mu\in\N$ to be specified later. For each $i\in[\mu]$, we define $\hat{\beta}_i = v\ \gamma_i\ v$ for $v = 0\#^30$ and $\gamma_i$ being some terminal-free pattern, also to be defined later. For each $i,j\in[\mu]$ with $i \neq j$ we assume by the following construction that $\var(\gamma_i) \cap \var(\gamma_j) = \varepsilon$, i.e., each variable occurring in one $\gamma_i$ does not occur anywhere else in the pattern. Also, for any word obtained by $\gamma_i$, we assume that it is either $0^{|\gamma_i|}$ or $0u0$ for some $u\in\Sigma^*$ by the construction that follows. Before giving the constraints $r_\beta$ and $\ell_\beta$, we continue with the construction of $\alpha$. We define $\alpha$ by
$$ \alpha = t\ v\ 0\ \alpha_1\ 0\ v\ t^2$$
for a new and independent variable $\alpha_1$, $v = 0\#^30$, and a word $t\in\Sigma^*$ that is obtained by the following morphism. Define $\psi : (\Sigma\times X)^* \rightarrow \Sigma^*$ by $\psi(0) = 0$, $\psi(\#) = \#$, and $\psi(x) = 0$ for all $x\in X$. Then, we say that $t = \psi(\beta)$. The constraints $r_\alpha$ and $\ell_\alpha$ are empty, hence $L_{NE}(\alpha,r_{\alpha},\ell_\alpha) = L_{NE}(\alpha) = \{tv\}\cdot\Sigma^+\cdot\{vt^2\}$. We proceed with the definition of the constraints $r_\beta$ and $\ell_\beta$. For each $i\in[\mu+1]$, we define the language of the variable $r_i$ by
$$ L(r_i) = \{0,\ \psi(r_i\hat{\beta}_ir_{i+1}...r_\mu\hat{\beta}_\mu r_{\mu+1}),\ t\psi(r_1\hat{\beta}_1r_2...r_{i-1}\hat{\beta}_{i-1}r_i)\}.$$
Thus, notice, that each $r_i$ may only be substituted to one of $3$ words: Either $0$, a specific suffix of $t$ or a specific prefix of $t^2$. Additionally, for all variables $x\in\var(\beta)$ in $\beta$ we add the constraint that $\#\notin L(x)$, i.e., no variable can be substituted to just the letter $\#$. Clearly, any language can be intersected with another regular language to obtain that constraint. The specific regular- and length-constraints for each $\gamma_i$ will be introduced later in the proof. For the length constraints $\ell_\beta$, we construct a disjunction of systems of linear (diophantine) inequalities. Assume w.l.o.g. that for each $i\in[\mu]$ we have $\var(\gamma_i) := \{x_{i,1},x_{i,2},...,x_{i,n_i}\}$ for some $n_i\in\N$. Then, we define for each $i\in[\mu]$, for each $j\in[n_i]$ of the resulting $n_i$ the system $\ell_{\beta,i,j}$ defined by:
\begin{itemize}
    \item $r_i = |\psi(r_i\hat{\beta_i}r_{i+1}...r_{\mu}\hat{\beta}_\mu r_{\mu+1})|$
    \item $r_{i+1} = |t\ \psi(r_1\hat{\beta}_1r_2...r_{i}\hat{\beta}_ir_{i+1})|$
    \item $r_k = 1$ for all $k\in[\mu+1]\setminus\{i,i+1\}$
    \item $x_{i,j} \geq 2$
    \item $x_{k,k'} = 1$ for all $k\in[\mu]$ and $x_k'\in\var(\gamma_k)$ for $k \neq i$
\end{itemize}
Additionally, for some $\ell_{\beta,i,j}$, there are additional length constraints on the specific variables in occurring in $\gamma_{i}$. These are defined later on in the proof and do not interfere with the constraint that $\gamma_i$ may be substituted to $0^{|\gamma_i|}$ in certain cases. We then define $\ell_\beta$ by $\ell_\beta = \ell_{\beta,1,1} \lor ... \lor \ell_{\beta,1,n_1} \lor \ell_{\beta,2,1} \lor ... \lor \ell_{\beta,2,n_2} \lor ...\ ... \lor \ell_{\beta,\mu_1,1} \lor ... \lor \ell_{\beta,\mu,n\mu}$. This constraint results in that exactly two subsequent $r_i$ and $r_{i+1}$ are substituted by words longer than the length of $1$ while all others have exactly length $1$. The other constraints result in that all variables occurring in some pattern $\gamma_j$ with $j\neq i$ have to be substituted by length $1$ and that at least one variable occurring in $\gamma_i$ has to be substituted to length $2$. These constraints serve selecting one $\gamma_i$ to obtain substitutions of $\alpha_1$ as seen in the following parts. We begin, again, with the important property of the containment of $L_{NE}(\beta,r_\beta,\ell_\beta)$ in $L_{NE}(\alpha,r_\alpha,\ell_\alpha)$.

\begin{lemma}
    We have $L_{NE}(\beta,r_\beta,\ell_\beta) \subseteq L_{NE}(\alpha,r_\alpha,\ell_\alpha)$.
\end{lemma}
\begin{proof}
    Let $w\in L_{NE}(\beta,r_\beta,\ell_\beta)$ such that $w = h(\beta)$ for some $h\in H_{r_\beta,\ell_\beta}$. We know by $\ell_\beta$ that there exists some $i\in[\mu]$ such that $|h(r_i)| = |\psi(r_i\beta_ir_{i+1}...r_\mu\hat{\beta}_\mu r_{\mu+1})|$, $|h(r_{i+1}) = |t\ \psi(r_1\hat{\beta}_1r_2...r_{i}\hat{\beta}_ir_{i+1})|$, and $|h(r_j)| = 1$ for all $j\in[\mu+1]\setminus\{i,i+1\}$. This, together with the languages defined by $r_\beta$, results in $h(r_i) =  \psi(r_i\hat{\beta}_ir_{i+1}...r_\mu\hat{\beta}_\mu r_{\mu+1})$, $h(r_{i+1}) = t\psi(r_1\hat{\beta}_1r_2...r_{i-1}\hat{\beta}_{i-1}r_i)$, and $h(r_j) = 0$ for all $j\in[\mu+1]\setminus\{i,i+1\}$. Also, by the constraint $|h(x_{j,k})| = 1$ for all $j\in[\mu]\setminus\{i\}$, $k\in\N$ given in $\ell_\beta$, we know that $h(x_{j,k}) = 0$ by $r_\beta$. Hence, all variables occurring outside $r_i\hat{\beta_i}r_{i+1}$ in $\beta$ are substituted by the word $0$. By that, we have $h(r_1\hat{\beta}_1r_2...r_{i-1}\hat{\beta}_{i-1}r_i) = t$ and $h(r_{i+1}\hat{\beta}_{i+1}r_{i+2}...r_\mu h(\beta)_\mu r_{\mu+1}) = t^2$. In total, this results in 
    \begin{align*}
        h(\beta) &= h(r_1\hat{\beta_1}r_2...r_{i-1\hat{\beta}_{i-1}r_i})h(\hat{\beta}_i)h(r_{i+1}\hat{\beta}_{i+1}r_{i+2}...r_\mu\hat{\beta}_\mu r_{mu+1}) \\
                 &= t\ v\ h(\gamma_i)\ v\ t^2
    \end{align*}
    Let $h'\in H$ and set $h'(0\alpha_10) = h(\gamma_i)$. This can be done by the assumption that all words obtained by $\gamma_i$ are either $0^{|\gamma_i|}$ or $0u0$ for some $u\in\Sigma^*$ and that there exists at least one variable in $\gamma_i$ which has to be substituted by length $2$, preventing the case of obtaining $0^{|\gamma_i|}$. Then, $h'(\alpha) = tvh'(\gamma_i)vt^2 = h(\beta)$ and by that $h(\beta)\in L_{NE}(\alpha,r_\alpha,\ell_\alpha)$.
\end{proof}

Next, we introduce the notion of predicates for this proof. Each $\gamma_i$ is interpreted as a predicate $\pi_i : H \rightarrow \{true,false\}$ for which we have that $\pi_i(h) = true$ for some $h\in h$ if and only if there exists some $r_\beta$-$\ell_\beta$-valid morphism $\tau$ for which we have that $\tau(\gamma_i) = h(0\alpha_10)$. Now, with that notion, we can show the following key-property.

\begin{lemma}\label{lemma:patlang-reglen-ne-predicates-iff}
    Let $h\in H$. We have $h(\alpha)\in L_{NE}(\beta)$ if and only if $\psi_i(h) = true$ for some $i\in[\mu]$, i.e., $h$ satisfies some predicate $\pi_i$.
\end{lemma}
\begin{proof}
    First, assume $\pi_i(h) = true$ for some $i\in[\mu]$. Then, there exists some $r_\beta$-$\ell_\beta$-valid $\tau$ such that $\tau(\gamma_i) = h(0\alpha_10)$. Let $h'$ be some substitution and set $h'(\gamma_i) = \tau(\gamma_i)$, $h'(r_i) = \psi(r_i\hat{\beta}_ir_{i+1}...r_\mu\hat{\beta}_\mu r_{\mu+1})$, $h'(r_{i+1}) = t\psi(r_1\hat{\beta}_1r_2...r_{i}\hat{\beta}_i r_{i+1})$, and $h'(x) = 0$ for all other variables $x\in\var(\beta)$. Then, notice that $h'$ is a $r_\beta$-$\ell_\beta$-valid substitution. We have
    \begin{align*}
        h'(\beta) &= h'(r_1\hat{\beta}_1r_2...r_{i-1}\hat{\beta}_{i-1})h'(r_i)h'(\hat{\beta}_i)h'(r_{i+1})h'(\hat{\beta}_{i+1}r_{i+1}...r_\mu\hat{\beta}_\mu r_{\mu+1}) \\
                  &= t\ h'(\hat{\beta}_i)\ t^2 \\
                  &= t\ v\ h'(\gamma_i)\ v\ t^2 \\
                  &= t\ v\ h(0\alpha_10)\ v\ t^2 \\
                  &= h(\alpha)
    \end{align*}
    and by that $h(\alpha)\in L_{NE}(\beta)$.

    Now, assume $h(\alpha)\in L_{NE}(\beta)$. Then, there exists $h'\in H_{r_\beta,\ell_\beta}$ such that $h'(\beta) = h(\alpha)$. We know by $\ell_\beta$ that there exists $i\in[\mu]$ such that $|h(r_i)| = |\psi(r_i\beta_ir_{i+1}...r_\mu\hat{\beta}_\mu r_{\mu+1})|$, $|h(r_{i+1}) = |t\ \psi(r_1\hat{\beta}_1r_2...r_{i}\hat{\beta}_ir_{i+1})|$, and $|h(r_j)| = 1$ for all $j\in[\mu+1]\setminus\{i,i+1\}$. This, together with the languages defined by $r_\beta$, results in $h(r_i) =  \psi(r_i\hat{\beta}_ir_{i+1}...r_\mu\hat{\beta}_\mu r_{\mu+1})$, $h(r_{i+1}) = t\psi(r_1\hat{\beta}_1r_2...r_{i-1}\hat{\beta}_{i-1}r_i)$, and $h(r_j) = 0$ for all $j\in[\mu+1]\setminus\{i,i+1\}$. We have
    \begin{align*}
        h(\alpha) &= t\ v\ 0\ h(\alpha_1)\ 0\ v\ t^2 \\
                  &= \psi(r_1\hat{\beta}_1r_2...r_{\mu}\hat{\beta}_{\mu}r_{\mu+1})\ v\ h'(\gamma_i)\ v\ t\ \psi(r_1\hat{\beta}_1r_2...r_{\mu}\hat{\beta}_{\mu}r_{\mu+1}) \\
                  &= \psi(r_1\hat{\beta}_1r_2...r_{i-1}\hat{\beta}_{i-1})h'(r_i)\ v\ h'(\gamma_i)\ v\ h'(r_{i+1})\psi(\hat{\beta}_{i+1}r_{i+2}...r_{\mu}\hat{\beta}_{\mu}r_{\mu+1}) \\
                  &= h'(r_1\hat{\beta}_1r_2...r_{i-1}\hat{\beta}_{i-1}r_i)\ v\ h'(\gamma_i)\ v\ h'(r_{i+1}\hat{\beta}_{i+1}r_{i+2}...r_{\mu}\hat{\beta}_{\mu}r_{\mu+1}) \\
                  &= h(\beta).
    \end{align*}
    By that, we directly see that we must have $h(0\alpha_1 0) = h(\gamma_i)$ and by that we know that $h$ satisfies $\pi_i$.
\end{proof}

With that, we have shown that for some word $h(\alpha)\in L_{NE}(\alpha,r_\alpha,\ell_\alpha)$ we have that it exists in $L_{NE}(\beta,r_\beta,\ell_\beta)$ if and only if there exist some predicate that is satisfied by $h$. Now, we give the specific construction of all those predicates $\gamma_i$ for $i\in[\mu]$ and show that those suffice to characterize all words what are not valid encodings of computations of $A$. By that, we obtain that $L_{NE}(\alpha,r_\alpha,\ell_\alpha) = L_{NE}(\beta,r_\beta,\ell_\beta)$ if and only if $\ValC(A) = \emptyset$. Notice, that the construction of these predicates is similar to the construction done in \cite{Nowotka2024}, just with slight adaptations to work in the non-erasing case. Additionally, the predicates handling invalid counter changes or invalid transitions could be simplified in this considered case.

We define the language of \emph{words of good structure} by $L_G = L((\#\#0^+\#00^+\#00^+)^+\#\#)$. A word is said to be of \emph{good structure} if $w\in L_G$. Otherwise, $w$ is of \emph{bad structure}. Notice that all encodings of computations of $A$ are words of good structure. From now on always assume $h\in H_{r_\alpha,\ell_\alpha}$ is a valid substitution for $\alpha$.

(1) The first predicate $\gamma_1$ is supposed to be used to obtain all words where $h(\alpha_1)$ is not of good structure or does not start with the encoding of the initial configuration $(q_0,0,0)$. Let $\gamma_1 = y_1$ for a new and independent variable $y_1\in X$ and set
$$ L(y_1) = \{0\} \cup (\{0\} \cdot L_{gs} \cdot \{0\}) $$
with
$$L_{gs} = \overline{L(\#\#0\#0\#0(\#\#0^+\#0^+\#0^+)^*\#\#)}. $$
If for some word $h(\alpha_1) = u$ we have that $u$ is not of good structure or does not start with the encoding of a computation with its first configuration being the initial configuration $\enc((q_0,0,0)) = 0\#0\#0$, then there exists some $r_\beta$-$\ell_\beta$-valid substitution $\tau$ such that $\tau(\gamma_1) = 0u0 = h(0\alpha_10)$, hence $\pi_1$ is satisfied. Otherwise, we can also substitute $\gamma_1$ to $0$, satisfying both constraints we set for all $\gamma$'s in the beginning.

(2) The second kind of predicates $\gamma_2$ to $\gamma_{|Q|-|F|+1}$ can be used to obtain all words that end with the suffix of a word of good structure where the last encoded configuration is not in a final state. We know there are $|Q|-|F|$ non-final states. Hence, for each $i\in[|Q|-|F|+1]\setminus[1]$ we define $\gamma_i = y_i$ for a new and independent variable $y_i\in X$ and set
$$ L(y_i) = \{0\} \cup (\{0\}\cdot\Sigma^*\cdot\{\ \#\#0^{1+j}\#0^{1+m_1}\#0^{1+m_2}\#\#0\mid q_j\in Q\setminus F, m_1,m_2\in\N\ \}). $$
Then, if for any $h(\alpha_1) = u$ we have that $u$ ends in an encoding of a configuration of $A$ which has a non-final state $q_j$, then we can use $\gamma_i$ to obtain a substitution $\tau$ such that $\tau(\gamma_i) = 0u0$, hence satisfying $\pi_i$. Also, we can obtain $\tau(\gamma_i) = 0$ by the definition of $\gamma_i$, also satisfying the second condition.

(3) By now, we handle all cases of bad structure, a wrong initial configuration, and an invalid final configuration. Next, we have to define $4$ predicates $\gamma_{|Q|-|F|+2}$ to $\gamma_{|Q|-|F|+5}$ that handle the case of invalid counter changes. In particular, if any of the counters $m_1$ or $m_2$ is either increased or decreased by more than $1$ in one step, the counter change is not valid and one of these $4$ predicates should handle that case. We construct the predicate $\gamma_{|Q|-|F|+2}$ that handles the case that the first counter is increased by more than one as an exemplary case. The constructions for the other $3$ predicates are analogue to this one. Notice that in this case, we also define length constraints for the variables added here. Let
$$ \gamma_{|Q|-|F|+2} = y_{|Q|-|F|+2,1}\ x_{|Q|-|F|+2,1}\ y_{|Q|-|F|+2,2}\ x_{|Q|-|F|+2,1}\ y_{|Q|-|F|+2,3}$$
for new and independent variables $y_{|Q|-|F|+2,1},y_{|Q|-|F|+2,2},y_{|Q|-|F|+2,3},x_{|Q|-|F|+2,1}\in X$ and set the regular constraints by
\begin{align*}
    L(y_{|Q|-|F|+2,1}) &:= \{0\} \cup\ \{\ 0u0\# \mid u\in\Sigma^*, |u|_{\#^3} = 0\ \}, \\
    L(y_{|Q|-|F|+2,2}) &:= \{0\} \cup\ L(\#0^+\#\#0^+\#\mathbf{0^2}), \\
    L(y_{|Q|-|F|+2,3}) &:= \{0\} \cup\ \{\ \#0u0 \mid u\in\Sigma^*, |u|_{\#^3} = 0\}, \\
    L(x_{|Q|-|F|+2,1}) &:= \{0\}^+.
\end{align*}
The bold letters mark the position where the property that two $0$'s are added to the counter is encoded. Additionally, we can define length constraints to simplify the cases we have to distinct. The additional length constraints $\ell_{\beta,|Q|-|F|+2,j}$ for all $j\in[|\var(\gamma_{|Q|-|F|+2,j})|]$ are given by the following disjunction of two systems:

\begin{itemize}
    \item $x = 1$ for all $x\in\var(\gamma_{|Q|-|F|+2})$
    \item[] $\qquad\qquad\qquad\lor$
    \item $x \geq 2$ for all $x\in\var(\gamma_{|Q|-|F|+2})\setminus\{x_{|Q|-|F|+2,1}\}$
\end{itemize}

Then, first of all, we can obtain $0^{|\gamma_i|} = 0^5$ from $\gamma_i$ by setting $\tau(x) = 0$ for all variables $x\in\var(\gamma_{|Q|-|F|+2})$. This satisfies the first constraint from before. Now, if and only if $h(0\alpha_10) = 0u0$ such that $u$ has a factor $\#0^{1+m_1}\#0^{1+m_2}\#\#0^{1+j}\#\mathbf{0^2}0^{1+m_1}\#$ for some $m_1,m_2\in\N$ and $q_j\in Q$, we find some $r_\beta$-$\ell_\beta$-valid substitution $\tau$ such that $\tau(\gamma_{|Q|-|F|+2}) = 0u0$. The only if direction follows by the definition of the regular constraints, the if direction follows by the fact, that all variables have to be either substituted to length $1$, resulting in $0^5$, or that every variable $y_{|Q|-|F|+2,k}$ for $k\in[3]$ has at least length $2$, resulting in a word with a factor as described above. By switching the position of the $\#\#$ and the additional $\mathbf{0^2}$ in $L(y_{|Q|-|F|+2,2})$, all other predicates for this case can be constructed analogously to this one.

(4) With $(1)$, $(2)$, and $(3)$ we have that some $h\in H_{\alpha}$ with $h(0\alpha_10) = 0u0$ does not satisfy any predicate, yet, if and only if $u$ is a word of good structure in which every subsequent encoded configuration has a counter change of at most $1$ per counter. By that, already all encodings of valid computations of $A$ are a subset of these words. However, the case of valid counter changes but invalid transitions is not handled yet. These are handled by the fourth and final kind of predicates. So, based on $\delta$ in $A$, for all $q_k,q_j\in Q$, $c_1,c_2\in\{0,1\}$, and $r_1,r_2\in\{-1,0,1\}$ with $(q_k,r_1,r_2)\notin\delta(q_j,c_1,c_2)$, we define a new and independent predicate $\gamma$ which can be used to obtain all encodings of computations that utilize this transition. The construction is demonstrated using an exemplary case and can be adapted, similarly to the predicates defined in (3), based on all other invalid transitions. Assume, for some invalid transition $(q_k,r_1,r_2)\notin\delta(q_j,c_1,c_2)$ we had $c_1 = 1$, $c_2 = 0$, $r_1 = +1$, $r_2 = 0$, and $q_j,q_k\in Q$. Assume $n$ is the index of the current predicate and set $\mu$ to be the total number of predicates overall. We set
$$ \gamma_n = y_{n,1}\ x_{n,1}\ y_{n,2}\ x_{n,2}\ y_{n,3}\ x_{n,1}\ y_{n,4}\ x_{n,2}\ y_{n,5}$$
for new and independent variables $y_{n,1},y_{n,2},y_{n,2},y_{n,4},y_{n,5},x_{n,1},x_{n,2}\in X$. The regular constraints for this predicate are defined by
\begin{align*}
    L(y_{n,1}) &:= \{0\} \cup \{\ u0\#\#0^{1+j}\# \mid u\in\{\varepsilon\}\cup (\{0\}\cdot\Sigma^*)\ \}, \\
    L(y_{n,2}) &:= \{0\} \cup \{\ 0\#\ \}, \\
    L(y_{n,3}) &:= \{0\} \cup \{\ \#\#0^{1+k}\#\ \},\\
    L(y_{n,4}) &:= \{0\} \cup \{\ 0\mathbf{0}\#\ \}, \\
    L(y_{n,5}) &:= \{0\} \cup \{\ \#\#0u \mid u\in\{\varepsilon\}\cup(\Sigma^*\cdot\{0\})\ \}, \\
    L(x_{n,1}) &:= \{0\}^+, \\
    L(x_{n,2}) &:= \{0\}^+.
\end{align*}
The bold letter marks the position where the $r_1 = +1$ is encoded. Again, we define length constraints to restrict the considered cases to be either that all variables are substituted to length $1$ or that all variables $y_{n,r}$ for $r\in[5]$ are substituted by at least length $2$. The additional length constraints for each $\ell_{\beta,n,r'}$ for $r'\in[|\var(\gamma_{n})|]$ are given by the following disjunction of two systems:

\begin{itemize}
    \item $x = 1$ for all $x\in\var(\gamma_n)$
    \item[] $\qquad\qquad\qquad\lor$
    \item $x \geq 2$ for all $x\in\var(\gamma_n)\setminus\{x_{n,1},x_{n,2}\}$
\end{itemize}

First, we can obtain $0^{|\gamma_{n}|} = 0^9$ from $\gamma_n$ by setting $\tau(x) = 0$ for all $x\in\var(\gamma_n)$. This satisfies the first assumption for $\gamma_n$ from before. Next, if and only if $h(0\alpha_10) = 0u0$ such that $u$ has a factor $$ \#\#0^{1+j}\#0^{2+m_1}\#0^{1+m_2}\#\#0^{1+k}\#0^{2+m_1}\textbf{0}\#0^{1+m_2}\#\# $$
for some $m_1,m_2\in\N$, we find some $r_\beta$-$\ell-\beta$-valid substitution $\tau$ such that $\tau(\gamma_n) = 0u0$. The only if direction follows by applying the definition of the regular constraints. The if direction follows by the fact that all variables are either substituted to length $1$, resulting in $0^9$, or that every variable $y_{n,r}$ for $r\in[5]$ has at least length 2, resulting in a word with the factor as described above. By switching the positions of $0$'s in $L(y_{n,2})$ and $L(y_{n,4})$ we can adapt this construction for all $c_1$, $c_2$, $r_1$, and $r_2$ values and construct all predicates for invalid transitions analogously. As there are only $2$ counters, $|Q|$ many states, two different values for $c_1$ and $c_2$ as well as three different values for $r_1$ and $r_2$, we know that the number of invalid transitions is finite. This concludes the construction of predicates of type (4).

Using all predicates defined in $(1)$, $(2)$, $(3)$, and $(4)$, given some $r_\alpha$-$\ell_\alpha$-valid $h\in H$ and applying Lemma \ref{lemma:patlang-reglen-ne-predicates-iff}, we can conclude that $h(\alpha)\notin L_{NE}(\beta,r_\beta,\ell_\beta)$ if and only if $h(\alpha) = 0u0$ and $u\in \ValC(A)$. As $L_{NE}(\beta,r_\beta,\ell_\beta) \subseteq L_{NE}(\alpha,r_\alpha,\ell_\alpha)$, we have that $L_{NE}(\alpha,r_\alpha,\ell_\alpha) = L_{NE}(\beta,r_\beta,\ell_\beta)$ if and only if $\ValC(A) = \emptyset$. This decides the emptiness problem of nondeterministic 2-counter automata without input, hence the equivalence problem for non-erasing pattern languages with regular- and length-constraints is undecidable.

Like in \cite{Nowotka2024}, for any larger alphabet, we may always restrict the alphabets of each variable using additional regular constraints. This allows for the reduction of the binary case to any larger alphabets. Additionally, also as shown in \cite{Nowotka2024}, we can use terminal-free patterns with additional regular constraints to obtain the exact same reduction by exchanging each terminal letter $0$ and $\#$ with an occurrence of a new variable $x_{'0'}$ and $x_{'\#'}$ respectively and setting the languages $L(x_{'0'}) = \{0\}$ and $L(x_{'\#'})$. Hence, this problem is undecidable for all alphabets $\Sigma$ of sizes $|\Sigma| \geq 2$ for all patterns, even if we are restricted to the terminal-free cases only. This concludes the proof of Theorem \ref{theorem:patlang-reglen-nonerasing-equivalence-undecidability}.

\newpage

\section{NP-hardness of the Unary Case for the Equivalence Problem for Non-Erasing Pattern Languages with Length Constraints}
\label{section:proof4}

    This section is dedicated to show Proposition~\ref{prop:pat-len-equiv-ne-hard}. W.l.o.g. assume $\Sigma = \{0\}$. We proceed with a reduction from $\mathtt{3SAT}$. Assume w.l.o.g. $\mathcal{X} = \{X_1,X_2,...\}$ to be a set of boolean variables and let $\overline{\mathcal{X}} = \{\overline{X_1}, \overline{X_2}, ...\}$ be the set of their negations. Let $\varphi = (\varphi_1,\varphi_2,...,\varphi_n)$ for $n\in\N$ with $n \geq 2$ be a $\mathtt{3SAT}$ formula in conjunctive normal form over $\mathcal{X}\cup\overline{\mathcal{X}}$ such that for each clause $\varphi_i$ with $i\in[n]$ we have w.l.o.g. $\varphi_i = (X_{i,1} \lor X_{i,2} \lor X_{i,3})$ with $X_{i,j}\in \mathcal{X}\cup\overline{\mathcal{X}}$ for $j\in\{1,2,3\}$. We define a function $f : \mathcal{X}\cup\overline{\mathcal{X}} \rightarrow X$ that maps each boolean variable to a pattern variable by
    $$ f(X) = \begin{cases}
                u_i    &\text{, if } x\in\mathcal{X} \text{ and } x = X_i \\
                v_i    &\text{, if } x\in\overline{\mathcal{X}} \text{ and } x = \overline{X_i}
              \end{cases}$$
    for new and independent pattern variables $u_i$ and $v_i$. Clearly, one of the two cases is always fulfilled for each boolean variable $X\in\mathcal{X}\cup\overline{\mathcal{X}}$. We proceed by defining the pattern $(\alpha,\ell_\alpha)$ by
    $$ \alpha = f(X_{1,1})f(X_{1,2})f(X_{1,3})y_1\ f(X_{2,1})f(X_{2,2})f(X_{2,3})y_2\ ...\ f(X_{n,1})f(X_{n,2})f(X_{n,3})y_n $$
    for new and independent variables $y_i\in X$ with $i\in[n]$. Notice that this pattern is terminal-free. The length constraints $\ell_\alpha$ are defined by the following system:
    \begin{itemize}
        \item $u_i + v_i = 3$ for all $i\in[n]$
        \item $y_i \leq 3$ for all $i\in[n]$
        \item $f(X_{i,1}) + f(X_{i,2}) + f(X_{i,3}) + y_i = 7$ for all $i\in[n]$
        \item $\sum_{i=1}^{n}(f(X_{i,1}) + f(X_{i,2}) + f(X_{i,3}) + y_i) = 7n$
    \end{itemize}
    Notice, that the final constraint results in that $L_{NE}(\alpha,\ell_\alpha)$ may only have the word $0^{7n}$ in it, as it restricts the length of all symbols in $\alpha$ together. Now, we set the second pattern with length constraints $(\beta,\ell_\beta)$ by
    $$ \beta = z$$
    for a new and independent variable $z\in X$ and define the length constraint $\ell_\beta$ by $z = 0^{7n}$. Hence, $L_{NE}(\beta,\ell_\beta) = \{0^{7n}\}$.
    
    To proof the reduction, first, assume there exists a satisfying assignment of variables $\phi$  for $\varphi$. Let $h$ be a substitution such that $h(u_i) = 00$ and $h(v_i) = 0$ if and only if $\phi(X_i) = true$ and $\phi(\overline{X_i}) = false$. Otherwise, set $h(u_i) = 0$ and $h(v_i) = 00$. As $\phi$ is satisfying $\varphi$, we know that $4 \leq |h(f(X_{i,1})f(X_{i,2})f(X_{i,3}))| \leq 6$ for all $i\in[n]$. Hence, we can always set $h(y_i) = 0^{7-|h(f(X_{i,1})f(X_{i,2})f(X_{i,3}))|}$ and obtain $|h(f(X_{i,1})f(X_{i,2})f(X_{i,3})y_i)| = 7$ for all $i\in[n]$. Hence, we $h$ must be $\ell_\alpha$-valid and we have $h(\alpha) = 0^7n$. As this is the only word we may obtain for $L_{NE}(\alpha,\ell_\alpha)$, we have $L_{NE}(\alpha,\ell_\alpha) = L_{NE}(\beta,\ell_\beta)$.

    For the other direction, assume $L_{NE}(\alpha,\ell_\alpha) = L_{NE}(\beta,\ell_\beta) = \{0^{7n}\}$. So, there exists some $h\in H_{\ell_\alpha}$ such that $h(\alpha) = 0^{7n}$. By $\ell_\alpha$, we know $h(f(X_{i,1})f(X_{i,2})f(X_{i,3})y_i) = 7$ for all $i\in[n]$. As $|h(y_i)| \leq 3$ must be the case, for each $i\in[n]$, there exists some $j\in\{1,2,3\}$ such that $|h(f(X_{i,j}))| > 1$, hence by $\ell_\alpha$, we must have $|h(f(X_{i,j}))| = 2$, resulting in $h(f(X_{i,j})) = 00$. We set an assignment of variables $\phi$ by $X_i = true$ and $\overline{X_i} = false$ if and only if $h(u_i) = 00$ and $h(x_i) = 0$, otherwise set $X_i = false$ and $\overline{X_i} = true$. As $|h(u_i)| + |h(v_i)| = 3$, we know such an assignment is a valid assignment for $\varphi$. As for each clause $\phi_i$ for $i\in[n]$ we find a variable that is set to $true$, we know that $\varphi$ is satisfied by $\phi$, concluding the reduction.

\end{document}